\documentclass[12pt,draftcls, onecolumn]{IEEEtran}

\usepackage{graphicx}
\usepackage{epsfig}
\usepackage{amsmath}
\usepackage{amsfonts}
\usepackage{enumerate}
\usepackage{multirow}
\usepackage{cite}
\usepackage{color}

\usepackage{lmodern,subfigure,multirow} %Type1-Schriftart f�r nicht-englische Texte
\usepackage{rotating}
\pdfoutput=1

\newtheorem{proposition}{Proposition}
\newtheorem{lemma}{Lemma}

\begin{document}

\title{Cross-layer Theoretical Analysis of NC-aided Cooperative ARQ Protocols in Correlated Shadowed Environments (Extended Version)}

%\author{\IEEEauthorblockN{Angelos Antonopoulos and Christos Verikoukis}\\
%\IEEEauthorblockA{Telecommunications Technological Centre of Catalonia (CTTC)\\
%Castelldefels, Barcelona, Spain\\
%Email: \{aantonopoulos, cveri\}@cttc.es}\\
%}

\author{Angelos~Antonopoulos,~\IEEEmembership{Member,~IEEE,} Aris~S.~Lalos,~\IEEEmembership{Member,~IEEE,}
 Marco~Di~Renzo,~\IEEEmembership{Member,~IEEE,} and~Christos~Verikoukis,~\IEEEmembership{~Senior Member,~IEEE}% <-this % stops a space
\IEEEcompsocitemizethanks{
\IEEEcompsocthanksitem A. Antonopoulos and C. Verikoukis are with the Telecommunications Technological Centre of Catalonia (CTTC), Castelldefels, Spain. %\protect\\
E-mail:\{aantonopoulos, cveri\}@cttc.es
\IEEEcompsocthanksitem A. S. Lalos is with the Department of Signal Theory and Communications (TSC) of the Technical University of Catalonia (UPC), Barcelona, Spain. %\protect\\
E-mail:  aristeidis.lalos@tsc.upc.edu
\IEEEcompsocthanksitem M. Di Renzo is with the Laboratory of Signals and Systems (L2S), University Paris-Sud, Paris, France. % \protect\\% <-this % stops a space
E-mail:  marco.direnzo@lss.supelec.fr
\IEEEcompsocthanksitem This paper has been peer-reviewed and provisionally accepted in IEEE Transactions on Vehicular Technology. Due to space limitations, the published version will not include the Appendices. The authors would like to express their gratitude to Prof. Chi Zhang (Associated Editor) and the anonymous reviewers for their constructive comments that contributed to the enhancement of this work.}
}

\maketitle              % typeset the title of the contribution

%\vspace{-40pt}

\begin{abstract}
%\boldmath
In this paper, we propose a cross-layer analytical model for the study of Network Coding (NC)-based Automatic Repeat reQuest (ARQ) Medium Access Control (MAC) protocols in correlated slow faded (shadowed) environments, where two end nodes are assisted by a cluster of relays to exchange data packets. The goal of our work is threefold: i) to provide general Physical (PHY) layer theoretical expressions for estimating crucial network parameters (i.e., network outage probability and expected size of the active relay set), applicable in two-way communications, ii) to demonstrate how these expressions are incorporated in theoretical models of the upper layers (i.e., MAC), and iii) to study the performance of a recently proposed NC-aided Cooperative ARQ (NCCARQ) MAC protocol under correlated shadowing conditions. Extensive Monte Carlo experiments have been carried out to validate the efficiency of the developed analytical model and to investigate the realistic performance of NCCARQ. Our results indicate that the number of active relays is independent of the shadowing correlation in the wireless links and reveal intriguing trade-offs between throughput and energy efficiency, highlighting the importance of cross-layer approaches for the assessment of cooperative MAC protocols.
\end{abstract}

% This preserv\cite{5288484}es the distinction between vectors and scalars. However,



\IEEEpeerreviewmaketitle

\section{Introduction}
\label{sec:intro}

The increasing density of wireless networks, due to the proliferation of mobile devices, leverages the deployment of cooperative systems, where the communication between a source and a destination takes place via intermediate relay nodes. The incorporation of multiple relays can lead to significant improvements, by appropriately exploiting the degrees of freedom that are introduced in the network. However, the fact that several relay nodes require simultaneous access to the channel stresses the need for new Medium Access Control (MAC) protocols for the effective relay coordination. The efficient MAC protocol design and assessment require the consideration of realistic physical (PHY) layer models and channel conditions (e.g., fast fading and shadowing), making imperative the need for a MAC/PHY cross-layer approach \cite{cl}.

Although the cross-layer concept was initially applied in conventional networks \cite{cl5,cl1,shad3,cl6}, its potential is also significant in cooperative scenarios, where the role of the PHY layer is even more pronounced, since the selection of the relay set and the need for cooperation are determined by the quality of the links between the communicating nodes. To that end, the authors in \cite{cross} propose a cross-layer theoretical model to analyze the performance of a cooperative wireless system that employs an Automatic Repeat reQuest (ARQ) mechanism for error control in fast fading environments. The same idea is extended in \cite{cl7}, where the authors present an analytical framework for studying the performance of reliable ARQ-based relaying schemes in multihop cooperative systems. The study in \cite{shad4} introduces a cross-layer analytical model for the assessment of a multi-relay cooperative ARQ MAC protocol by taking into account the shadowing effect. In \cite{cl8}, a cooperative cross-layer MAC protocol, which combines space-time coding and adaptive modulation at the PHY layer, is proposed and analyzed. More recently, the work published in \cite{cl2} studies fundamental cooperative issues (i.e., when and whom to cooperate with) from a cross-layer perspective in distributed wireless networks.

In addition to the one-way cooperative schemes, during the last few years, the implementation of new software applications, based on Voice over IP (VoIP) and instant messaging, has driven the need for two-way (bidirectional) communication, further complicating the design of effective cooperative systems. To deal with this new trend for bidirectional communication, Network Coding (NC) has been proposed as an alternative routing mechanism that enables the relays to mix the incoming data packets before forwarding them to their final destinations. Apparently, the application of NC implies straightforward gains in bidirectional networks, since the relay nodes require less resources for their transmissions. This potential advantage has lately inspired several works \cite{cope,argyriou,phoenix,wang,umehara}, focusing on the design of novel cooperative MAC protocols with NC capabilities to enhance the throughput, the energy efficiency and the robustness of wireless networks. In the same context, motivated by the great interest that ARQ schemes have attracted in the literature, we have introduced an NC-aided Cooperative ARQ-based MAC protocol \cite{nccarq}, namely NCCARQ, which exploits the benefits of both NC and ARQ to improve the performance of cooperative wireless networks.

Despite their inherent differences on the channel access rules, most NC-aided MAC protocols share the common assumption of either ideal channel conditions or simplified PHY layer models.  However, the existing cross-layer models for simple one-way cooperative networks do not apply directly in bidirectional communications, where the relays are selected according to the packets that have been received from both directions. In addition, another basic limitation of the existing models is the assumption of independent wireless links in the network, although recent studies \cite{cor1,cor2,cor3,cor4} have indicated the impact of shadowing spatial correlation (due to geographically proximate wireless links) on the performance of cooperative MAC protocols. Hence, considering the above limitations, the accurate performance evaluation of NC-aided protocols in correlated environments becomes essential for an efficient network planning, reducing the deployment and operational cost of the cooperative systems.

In this paper, taking into account the gaps in the current literature along with the importance of cross-layer modeling, we present a joint MAC/PHY theoretical framework to evaluate the throughput and the energy efficiency of NC-aided ARQ schemes under correlated shadowing conditions. Our main contributions can be summarized as follows:

\begin{enumerate}
\item We introduce a cross-layer analytical framework that jointly considers the MAC layer operation and the PHY layer conditions in NC-based communication scenarios. Without loss of generality, we use as an exemplary case the recently proposed NCCARQ MAC protocol \cite{nccarq} to study how correlated shadowing affects crucial protocol parameters.
\item We analytically demonstrate that the average number of active relays in the network is independent of the correlation among the wireless links from the end nodes to the relays.
\item We provide practical insights for efficient network planning for NC-based cooperative communications by revealing interesting tradeoffs between the throughput and energy efficiency performance in the network under realistic channel conditions.
\end{enumerate}

The remainder of this paper is organized as follows. Section \ref{sec:system} presents our system model, focusing on a two-way communication scenario with correlated wireless links. Section \ref{sec:impact} provides an overview for NCCARQ, highlighting the impact of the PHY layer on the protocol design and performance. In Section \ref{sec:analysis}, we introduce a joint MAC/PHY analytical framework for the throughput and the energy efficiency of the network. The validation of the model and the performance evaluation of the protocol under correlated shadowing conditions are provided in Section \ref{sec:performance}. Finally, Section \ref{sec:conclusions} concludes the paper.

\section{System Description}
\label{sec:system}

\subsection{Channel Model}
\label{sec:channel}

The network under consideration (Fig. \ref{f1}) consists of two end nodes ($A$ and $B$) that have data packets to exchange in a bidirectional communication, and a set of $n$ intermediate nodes ($R_1,R_2,...R_n$) with NC capabilities that act as relays in this network setup, assisting the communication towards both directions. The instantaneous received power at any given node $j$ from transmissions by node $i$ is denoted by $\gamma_{ij}=\frac{P_{Tx}}{d_{ij}^a} \left|h_{f_{ij}}\right|^2\left|h_{s_{ij}}\right|^2$ \cite[Eq. (1.1)]{thesismary}, where: i) $P_{Tx}$ is the common transmission power for all nodes in the network, ii) $d_{ij}$ is the $(i,j)$ distance, iii) $a$ is the path-loss coefficient, iv) $h_{f_{ij}}$ is the fast fading coefficient, modeled as a Nakagami-m random variable (RV) with $\mathbf{E}\left[\left|h_{f_{ij}}\right|^2\right]=1$, and v) $h_{s_{ij}}$ is the shadowing coefficient.

\begin{figure}[htb]
\centering
\includegraphics[width=1\columnwidth]{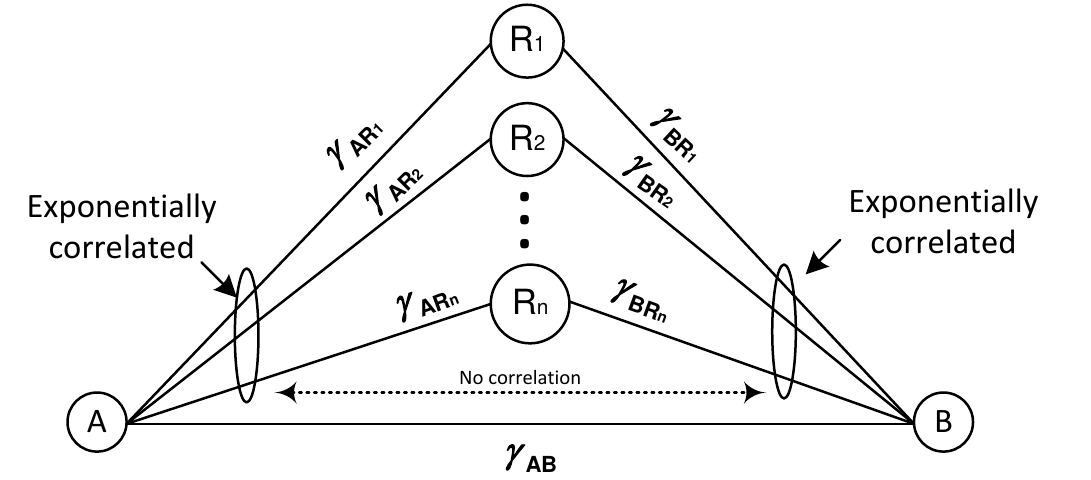}
\caption{System Model}\label{f1}
\end{figure}

With regard to the channel coefficients, the fast fading ergodicity allows the calculation of its mean value from a sufficiently long sample realization of the process (e.g., data packet duration). On the other hand, shadowing is a slowly varying procedure and, thus, it can be considered unaltered for the same or even larger period of time \cite{thesismary}. In our work, we assume that shadowing remains constant during a communication round, which consists of the direct transmission and the cooperation phase. Therefore, since the analysis is performed in packet-level, the average received power computed over the duration of one packet may be written as $\bar{\gamma}_{ij}$=$\mathbf{E}\left[\frac{P_{Tx}}{d_{ij}^a}\left|h_{s_{ij}}\right|^2\left|h_{f_{ij}}\right|^2\right]$ = $\frac{P_{Tx}}{d_{ij}^a}\left|h_{s_{ij}}\right|^2$. According to several experimental studies (e.g., \cite{5288484}), $h_{s_{ij}}$ and, consequently, $\bar{\gamma}_{ij}$ can be modeled as a log-normal RV, which implies that $\bar{\gamma}_{ij_{dB}} = 10 log_{10}\left(\bar{\gamma}_{ij}\right)$ is a normally distributed RV with mean value $\mu_{ij_{dB}}$ and standard deviation $\sigma_{ij_{dB}}$\footnote{In the rest of the paper, for the sake of clarity and without loss of generality, the values of $\gamma, \mu, \sigma$ are always expressed in dB.}.

Regarding the correlation model, we denote by $\rho_{1_{x,y}}$ the correlation factor between two links $AR_{x}$ and $AR_{y}$, and by $\rho_{2_{x,y}}$ the correlation factor between two links $BR_{x}$ and $BR_{y}$, respectively. On the other hand, no correlation is assumed between $AR_{x}$ and $BR_{y}$ links, i.e., $\rho(\bar{\gamma}_{AR_{x}},\bar{\gamma}_{BR_{y}}) = 0, \forall\ x,y$. The correlation factors $\rho_{1_{x,y}}$ and $\rho_{2_{x,y}}$ can be estimated as:

\begin{equation}
\rho_{1_{x,y}}=\rho\left(\bar{\gamma}_{AR_{x}},\bar{\gamma}_{AR_{y}}\right)=\mathbf{E}\left[\left(\bar{\gamma}_{AR_{x}} - \mu_{AR_{x}}\right)\left(\bar{\gamma}_{AR_{y}} - \mu_{AR_{y}}\right)\right]/\sigma_{AR_{x}}\sigma_{AR_{y}}, \forall x,y\in [1,n]
\end{equation}
\begin{equation}
\rho_{2_{x,y}}=\rho\left(\bar{\gamma}_{BR_{x}},\bar{\gamma}_{BR_{y}}\right)=\mathbf{E}\left[\left(\bar{\gamma}_{BR_{x}} - \mu_{BR_{x}}\right)\left(\bar{\gamma}_{BR_{y}} - \mu_{BR_{y}}\right)\right]/\sigma_{BR_{x}}\sigma_{BR_{y}}, \forall x,y\in [1,n].
\end{equation}

In addition, taking into account that the links to each direction have a common end point, we assume that the correlation between any pair of links $\rho_{1_{x,y}}$ or $\rho_{2_{x,y}}$  decreases exponentially as the distance between them increases, i.e., $\rho_{x,y} = \rho^{\left|x-y\right|}$ where $\rho \in [0,1]$\cite{104090}. To that end, a set of exponentially correlated normal RVs $\boldsymbol{\gamma}_{AR_{x}}=\left[\bar{\gamma}_{AR_{1}}, \ldots, \bar{\gamma}_{AR_{n}}\right]$ can be generated as:

\begingroup
\begin{equation}
\boldsymbol{\gamma}_{AR_{x}} = \mathbf{\sigma}_1 \left(\mathbf{\Sigma_n}\left(\rho_1\right)\right)^{1/2}\textbf{X}_{n\times 1}+\mathbf{\mu}_1, \label{eq1}
\end{equation}
\endgroup

\noindent where $\textbf{X}_{n\times 1} = \left[X_1,\ldots,X_n\right]^T$, with $X_i\sim \mathcal{N}(0,1)$, $\mathbf{\mu}_1 = \left[\mu_{AR_{1}},\ldots,\mu_{AR_{n}}\right]^T$, $\mathbf{\sigma}_1$ is a diagonal matrix that contains the $\sigma_{AR_{i}}$ values in its main diagonal, i.e., $\mathbf{\sigma}_1 = diag\{\sigma_{AR_{1}},\ldots,\sigma_{AR_{n}}\}$, while $\mathbf{\Sigma_n}\left(\rho_1\right)$ can be expressed as a Toeplitz matrix\footnote{This matrix is also known as Kac-Murdock-Szeg\"{o} matrix \cite{toeplitz}.}, whose entries depend on the correlation factor $\rho_1$:

\begingroup
\begin{equation}
\mathbf{\Sigma_n}\left(\rho_1\right) = \left[\begin{array}{ccccc}
1&\rho_1&\rho_1^2&\cdots&\rho_1^n\\
\rho_1&1&\rho_1&\cdots&\rho_1^{n-1}\\
\rho_1^2&\rho_1&1&\cdots&\rho_1^{n-2}\\
\vdots&\ddots&\ddots&\ddots&\vdots\\
\rho_1^n&\rho_1^{n-1}&\cdots&\rho_1&1
\end{array}\right].
\end{equation}
\endgroup

\noindent Accordingly, the exponentially correlated normal RVs $\boldsymbol{\gamma}_{BR_{x}}=\left[\bar{\gamma}_{BR_{1}}, \ldots, \bar{\gamma}_{BR_{n}}\right]$ can be generated as:

\begingroup
\begin{equation}
\boldsymbol{\gamma}_{BR_{x}} = \mathbf{\sigma}_2 \left(\mathbf{\Sigma_n}^{1/2}\left(\rho_2\right)\right)\textbf{Y}_{n\times 1}+\mathbf{\mu}_2, \label{eq2}
\end{equation}
\endgroup

\noindent where $\textbf{Y}_{n\times 1} = \left[Y_1,\ldots,Y_n\right]^T$, with $Y_i\sim \mathcal{N}(0,1)$, $\mathbf{\mu}_2 = \left[\mu_{BR_{1}},\ldots,\mu_{BR_{n}}\right]^T$, $\mathbf{\sigma}_2=diag\left\{\sigma_{BR_{1}},\ldots,\sigma_{BR_{n}}\right\}$ and $\mathbf{\Sigma_n}\left(\rho_2\right)$ is a Toeplitz matrix, function of the correlation factor $\rho_2$.

We further assume that node $B$ is marginally located in the transmission range of node $A$ (and vice versa), which implies a weak direct link with relatively low $\bar{\gamma}_{AB}$. However, the erroneous direct transmissions are compensated by employing network cooperation through ARQ control mechanisms.

\subsection{Packet Acceptance Criteria}
\label{sec:criteria}

In wireless networks, different applications (e.g., video, gaming, e-mail, etc.) require different levels of QoS, which can be provisioned through a target Packet Error Rate (PER) denoted by the probability $p^*$. Therefore, metrics such as the Average PER (APER) or the Outage PER (OPER) have to be employed in order to determine the correct reception of a packet according to the target value of $p^*$. In our case, a given relay should receive correct packets by both $A$ and $B$ in order to be able to apply NC and participate in the cooperation phase. As a result, the realistic channel conditions affect: i) the size of the active relay set ($\mathcal{A}_n$), which is composed of the relays that successfully receive packets from both end nodes, and ii) the network outage probability ($p_{out}$), defined as the probability that none of the available $n$ relays in the system receives both packets successfully, as it is assumed that the shadowing coefficients remain constant during one communication round.

In this point, let us focus on the metrics that are used to verify the correct packet reception under fast and slow fading conditions. In environments where shadowing is not considered, the ergodicity of fast fading allows the utilization of average metrics, such as the APER, to characterize the system performance and determine the acceptance of a packet. Thus, under fast fading conditions for a given PHY layer set up, the APER between two nodes $i$ and $j$ increases monotonically with their distance, i.e., $APER_{ij} = f(d_{ij})$ \cite{5288484}.

On the other hand, the criterion of correct packet reception is substantially modified in the presence of slow fading, which is a non-ergodic process. As we have seen in Section~\ref{sec:channel}, the received power $\bar{\gamma}_{ij}$ eventually depends only on the $h_{s_{ij}}$ coefficient, since $h_{f_{ij}}$ can be averaged because of its ergodicity. This implies that, in slow fading environments, the APER is a function of distance and shadowing (i.e., $APER_{ij} = f(d_{ij},h_{s_{ij}})$), and the QoS requirement $APER_{ij}\leq p^*$ is equivalent to $\bar{\gamma}_{ij}>\gamma^*$ \cite{5288484}. However, although shadowing is a non-ergodic process, $\bar{\gamma}_{ij}$ is still an RV that requires statistical characterization. In this case, the most suitable metric is the OPER (i.e., the probability of receiving an erroneous packet) and the normal distribution of $\bar{\gamma}_{ij}$ (in dB) allows us to express it as:

\begin{equation}
OPER_{ij}=Pr\left\{APER_{ij}> p^*\right\}=Pr\left\{\bar{\gamma}_{{ij}}\leq \gamma^*\right\}=1-Q\left(\frac{\gamma^*-\mu_{{ij}}}{\sigma_{{ij}}}\right),
\end{equation}
where $Q\left(\cdot\right)$ is the standard one dimensional Gaussian Q-function, traditionally defined by $Q\left(x\right) = \int^{\infty}_{x}\frac{1}{\sqrt{2\pi}}e^{\frac{-t^2}{2}}dt$. The above expression suggests as sufficient and necessary condition for the packet acceptance that the mean received power should be above the threshold value $\gamma^*$. Mary et al. \cite{5288484} have provided closed-form formulas for $\gamma^*$ as a function of a target symbol error probability set by the application layer for log-normal shadowing and Nakagami-m wireless channels.

\section{NCCARQ Overview and PHY Layer Impact}
\label{sec:impact}

The goal of this section is to highlight the impact of realistic PHY layer on the performance of NC-aided MAC protocols. To that end, we use as a representative case study the NCCARQ MAC protocol \cite{nccarq}, which coordinates the channel access among a set of NC-capable relay nodes in a bidirectional wireless communication. In the following sections, we briefly review the protocol's operation and we explicitly study the changes due to the realistic PHY layer consideration.

\subsection{NCCARQ Overview}
\label{sec:overview}

NCCARQ \cite{nccarq} MAC protocol has been designed to exploit the benefits of both ARQ and NC in two-way cooperative wireless networks, being backwards compatible with the Distributed Coordination Function (DCF) of the IEEE 802.11 Standard \cite{80211}. The function of the protocol is based on two main factors: i) the broadcast nature of wireless communications, which enables the cooperation between the mobile nodes, and ii) the capability of the intermediate relay nodes to perform NC before any transmission.

Fig. \ref{nccmac} presents an example of the frame sequence in NCCARQ, where two end nodes ($A$ and $B$) want to exchange their data packets (a and b, respectively) with the assistance of three NC-capable relay nodes ($R_1,R_2,R_3$). In this particular example, the protocol operates as follows:

\begin{itemize}
  \item Node A transmits packet a to node B. The relays overhear the transmission correctly, while we assume that node B fails to demodulate the received packet.
  \item Node B triggers the cooperation phase by broadcasting a Request For Cooperation (RFC) control packet. In addition, unlike conventional cooperative ARQ protocols, NCCARQ allows piggyback data transmissions along with the RFC (in this example, the data packet b), thus leveraging the NC application.
  \item After the reception of the RFC and since we assume ideal channel conditions, the relays apply NC to the two data packets (a and b) and set up their backoff counters according to the DCF rules in order to gain channel access and transmit the NC packet (a$\oplus$b) to the end nodes.
  \item In this example, we assume that the three relays select the values $R_1=2$, $R_2=2$ and $R_3=3$ for their backoff counters, respectively. As a result, after two time slots, $R_1$ and $R_2$ attempt a concurrent transmission and $R_3$ freezes its counter.
  \item The simultaneous packet transmission results in a collision and, according to the DCF rules, the two relays reset their backoff counters to $R_1=5$ and $R_2=12$, respectively, while $R_3=1$. Therefore, after one time slot, $R_3$ transmits the coded packet and the two destinations sequentially broadcast acknowledgment (ACK) packets, terminating the cooperation phase.
\end{itemize}

\begin{figure}[htb]
\centering
\includegraphics[width=1\columnwidth]{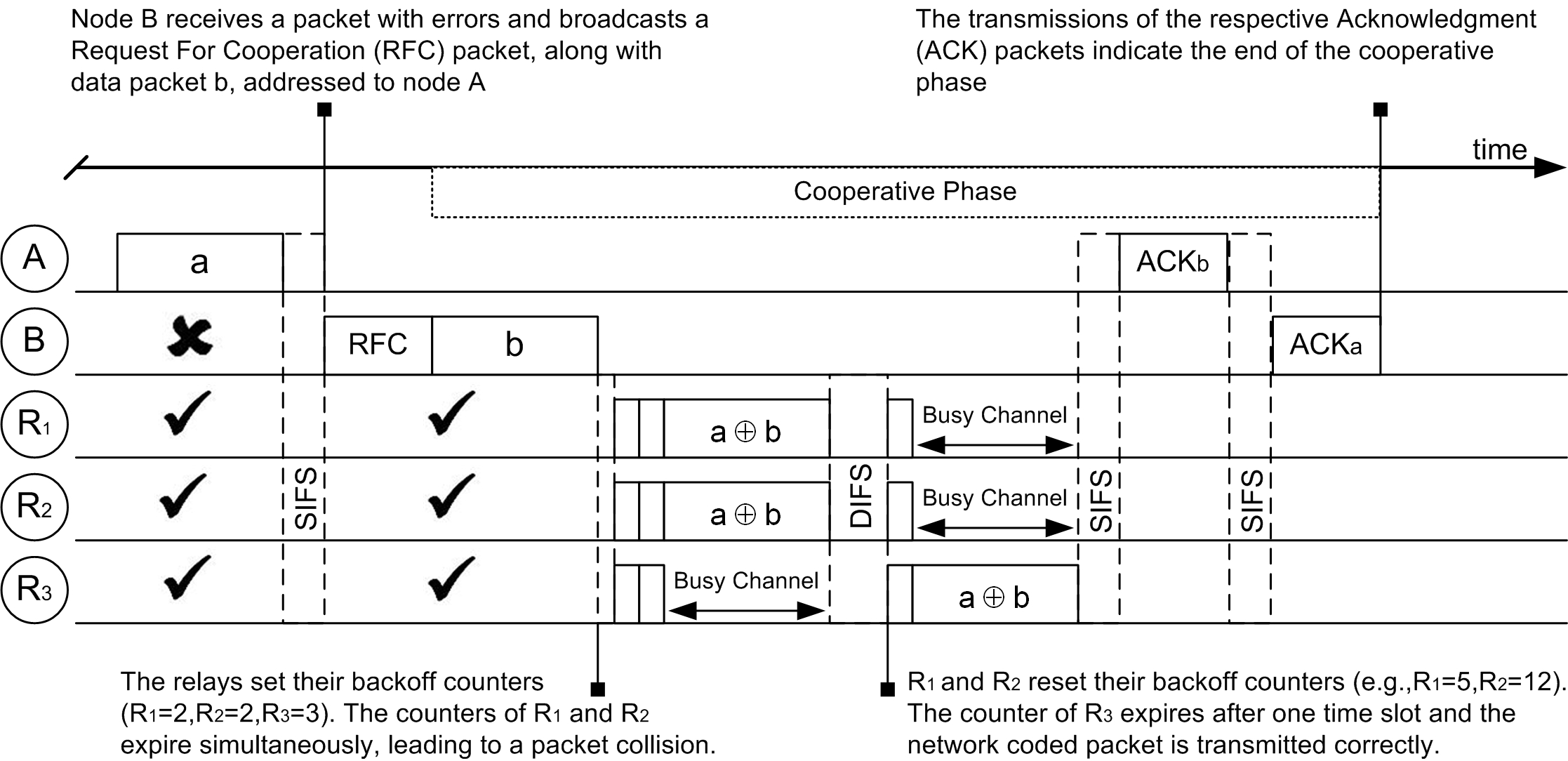}
\caption{NCCARQ operation without PHY layer consideration}\label{nccmac}
\end{figure}

Hence, the participation of multiple nodes in the contention phase results in idle slots and collisions in the network, before eventually a relay node manages to successfully transmit the coded packet. However, apart from the collisions and the idle periods, the protocol performance may be also degraded due to fading (either fast or slow) introduced by taking into account non-ideal channel conditions. In the next section, we provide some insights for the modifications that the realistic PHY layer potentially brings to the protocol operation.

\subsection{PHY Layer Impact}
\label{sec:phy}

The PHY layer consideration significantly modifies the protocol operation, as it is depicted in Fig. \ref{nccphy}. In this case, the protocol operates as follows:
\begin{itemize}
  \item Node A transmits packet a to node B. Node B and $R_3$ fail to demodulate the received packet, while $R_1$ and $R_2$ overhear the transmission correctly.
  \item Node B triggers the cooperation phase by broadcasting an RFC control packet along with data packet b. In this example, we assume that only $R_3$ receives correctly the data packet.
  \item Since $R_1$ and $R_2$ have received only packet a, and $R_3$ has received only packet b, there is no node in the relay set that can apply NC. As a result, after a predefined time ($T_{timeout}$), node A starts a new communication by transmitting packet a, which is correctly received by $R_1$ and $R_3$.
  \item Node B broadcasts again an RFC control packet along with the data packet b, which is correctly received by all the relays.
  \item Consequently, in this communication round, $R_1$ and $R_3$ have correctly received both packets, thus being able to participate in the cooperation phase. Accordingly, they set up their backoff counters to $R_1=2$ and $R_3=3$, respectively, and $R_2$ gains access to the channel after two time slots.
  \item The two destinations receive correctly the coded packet and they are able to extract the original packets a and b, terminating the cooperation phase by transmitting the respective ACK packets.
\end{itemize}

Apparently, the correct packet transmissions define the active relay set ($\mathcal{A}_n$), introducing the concept of a node being in outage. Hence, in the extreme case where no relay node has received both packets from $A$ and $B$, the relay set is in outage and the cooperation phase ends after a predefined time ($T_{timeout}$), which is not considered in systems that operate under ideal channel conditions. On the other hand, the reduction of the active relay set due to non-successful packet receptions could be beneficial in networks with many relays, since a smaller number of active relays would lead to a lower packet collision probability in the network. Hence, the aforementioned issues stress the necessity for designing accurate cross-layer models that consider the protocol operation in realistic conditions.

\begin{figure}[htb]
\centering
\includegraphics[width=1\columnwidth]{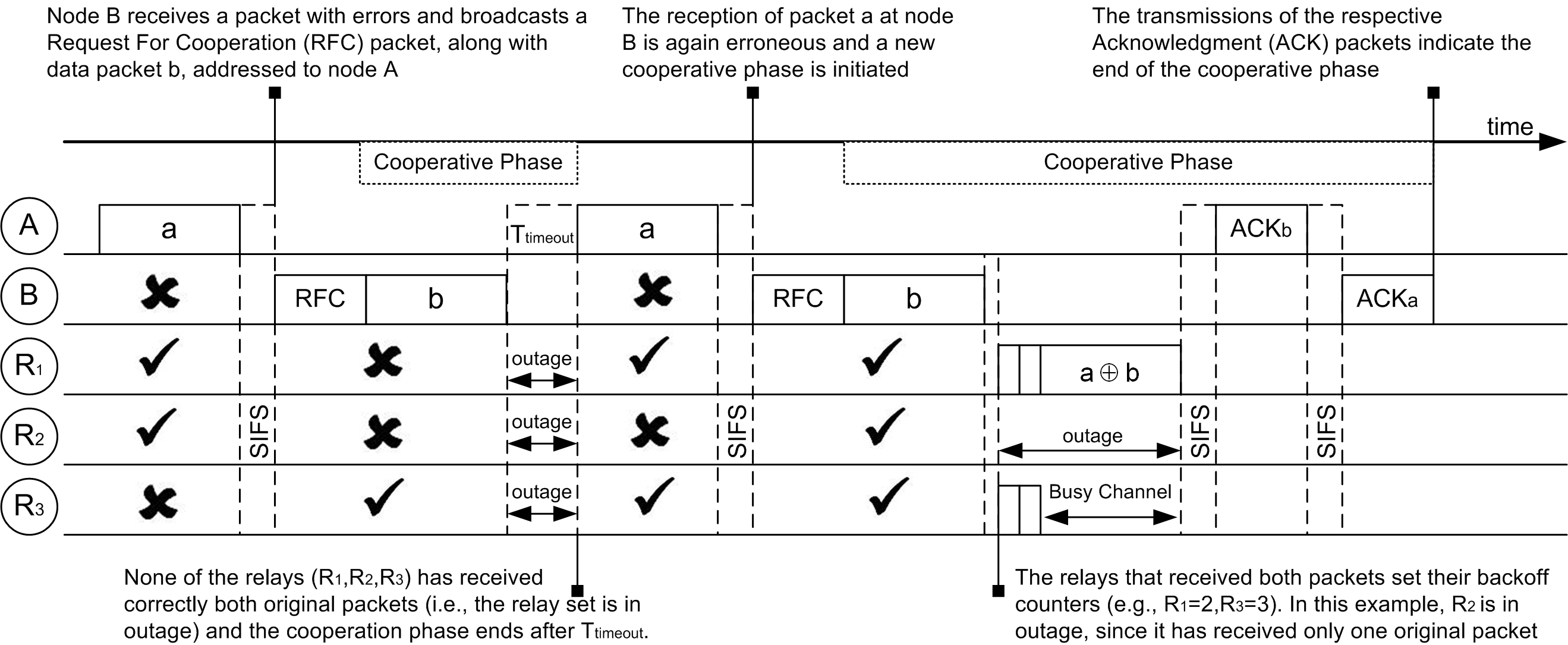}
\caption{NCCARQ operation with PHY layer consideration}\label{nccphy}
\end{figure}

\section{Joint MAC/PHY Analytical Framework}
\label{sec:analysis}

In this section, we introduce a joint MAC/PHY analytical framework to model the throughput and the energy efficiency achieved by NCCARQ under correlated shadowing and fast fading conditions. Although a complete analysis from the MAC layer point of view under ideal channel conditions is presented in \cite{nccarq}, the PHY layer consideration introduces new challenges in the theoretical derivations. In particular, concepts such as the network outage probability ($p_{out}$), the expected size of the active relay set ($\mathbf{E}\left[\left|\mathcal{A}_n\right|\right]$) and the OPER should be explicitly considered for an accurate analytical design. In the remainder of this section, we first focus on the parameters that are affected by the realistic PHY layer assumption and, then, we appropriately incorporate these parameters in a modified analysis from the MAC layer point of view.

\subsection{Physical Layer Impact on $p_{out}$ and $\mathbf{E}\left[\left|\mathcal{A}_n\right|\right]$}
\label{sec3}

The probability of having exactly $k$ active out of $n$ total relays in the system (i.e., $Pr\left\{\left|\mathcal{A}_n\right|=k\right\}$) is a required parameter for the estimation of both the network outage probability and the expected size of the active relay set. To that end, let us define by $\boldsymbol{1}_{AR_i} = \left\{\bar{\gamma}_{{AR_i}}>\gamma^*\right\}$ and $\boldsymbol{0}_{AR_i} = \left\{\bar{\gamma}_{{AR_i}}\leq\gamma^*\right\}$ the events that the relay $i$ receives from node $A$ a ``correct" or an ``erroneous" packet, respectively. In addition, we introduce the notation $\mathfrak{b}_{A\chi_n}$ to identify which of the $n$ relays have correctly received packets transmitted by node $A$. In particular, $\chi\in [0,2^n-1]$ is a natural number, whose value is specified by the combination of accepted and discarded packets in all $AR_i$ links, while $\mathfrak{b}_{A\chi_n}$ corresponds to the $n$-bit representation of $\chi$, where the positions of 1s indicate the specific relays in which the average received power $\bar{\gamma}_{AR_i}$ is above the reliability threshold $\gamma^*$\footnote{For example, $\mathfrak{b}_{A5_3}=[\boldsymbol{1}_{AR_1},\boldsymbol{0}_{AR_2},\boldsymbol{1}_{AR_3}]$ indicates that: i) there are 3 relays in the network ($R_1,R_2,R_3$), and ii) only $R_1$ and $R_3$ have received correct packets from node A.}. Accordingly, $\mathfrak{b}_{B\psi_n}$ identifies which relays have successfully received packets transmitted by node $B$, where $\psi$ has the same characteristics as $\chi$.

Hence, the probability that exactly $k$ out of $n$ relays have successfully received packets from both $A$ and $B$ (i.e., $Pr\left\{\left|\mathcal{A}_n\right|=k\right\}$) may be estimated by taking into account all the possible binary codewords $\mathfrak{b}_{A\chi_n}$ and $\mathfrak{b}_{B\psi_n}$ that satisfy the following condition:

\begin{equation}
H_w\left(\mathfrak{b}_{A\chi_n}\odot \mathfrak{b}_{B\psi_n}\right)=k,
\end{equation}
where $\odot$ denotes the bit wise AND operation, while $H_w\left(\mathfrak{b}\right)$ corresponds to the Hamming weight function that returns the number of 1s in the binary word $\mathfrak{b}$. Denoting as $Pr\{\mathfrak{b}_{A\chi_n}\}$ and $Pr\{\mathfrak{b}_{B\psi_n}\}$ the probability of occurrence of each possible event $\mathfrak{b}_{A\chi_n}$ and $\mathfrak{b}_{B\psi_n}$, respectively, the aforementioned probability is given by:

\vspace{-1pt}

\begingroup
\begin{equation}
Pr\left\{\left|\mathcal{A}_n\right|=k\right\} = \mathop{\sum_{\chi=0}^{2^n-1}\sum_{\psi=0}^{2^n-1}}_{\{H_w\left(\mathfrak{b}_{A\chi_n}\odot \mathfrak{b}_{B\psi_n}\right) = k\}}Pr\{\mathfrak{b}_{A\chi_n}\}Pr\{\mathfrak{b}_{B\psi_n}\}. \label{P-k-n-2}
\end{equation}
\endgroup

\subsubsection{Theoretical Estimation of the Probabilities $Pr\{\mathfrak{b}_{A\chi_n}\}$ and $Pr\{\mathfrak{b}_{B\psi_n}\}$}
In order to further clarify the concepts and derive expressions for the probabilities $Pr\{\mathfrak{b}_{A\chi_n}\}$ and $Pr\{\mathfrak{b}_{B\psi_n}\}$, let us provide an example for the theoretical estimation of $Pr\{\mathfrak{b}_{A1_n}\}$, which corresponds to the probability that only relay $n$ receives a ``correct" packet from node A, while all the other relays (i.e., $R_1,\ldots, R_{n-1}$) receive ``erroneous" packets:

\begingroup
\begin{eqnarray}
&&Pr\{\mathfrak{b}_{A1_n}\} = Pr\left\{\boldsymbol{0}_{AR_1},\boldsymbol{0}_{AR_2},...,\boldsymbol{1}_{AR_n}\right\}\nonumber\\
&&=Pr\left\{\bar{\gamma}_{{AR_1}}\leq \gamma^*,\bar{\gamma}_{{AR_2}}\leq \gamma^*....,\bar{\gamma}_{{AR_n}}> \gamma^*\right\}\nonumber\\
&&=  \int^{\frac{\gamma^*-\mu_{AR_1}}{\sigma_{AR_1}}}_{-\infty}\cdots\int^{\infty}_{\frac{\gamma^*-\mu_{AR_n}}{\sigma_{AR_n}}} f_y\left(y_1, \ldots, y_n\right) dy_1\cdots dy_n,\label{eq:exp-int}
%&&= C_0\int^{\frac{\gamma_{0_{dB}}-\mu_{AR_{1}}}{\sigma_{dB}}}_{-\infty}exp\left[-\frac{ y^2_N}{2\left(1-\rho^2\right)}\right]q_{N-1}\left(y_N\right)dy_N\label{eq:exp-int}
\end{eqnarray}
\endgroup
where $f_y\left(y_1, \ldots, y_n\right)$ corresponds to the joint Probability Density Function (PDF) of the RVs $y_i = \left(\bar{\gamma}_{{AR_1}} -\mu_{AR_1}\right)/ \sigma_{AR_1}$ and can be written as:
\begin{equation}
f_y\left(y_1, \ldots, y_n\right) = f_y\left(\mathbf{y}\right) = [det\left(\mathbf{\Sigma_n}\left(\rho_1\right)\right)]^{-1/2}\left(2\pi\right)^{-n/2}exp\left(-\frac{\mathbf{y}^T\mathbf{\Sigma_n}^{-1}\left(\rho_1\right)\mathbf{y}}{2}\right),\label{jointpdf}
\end{equation}
where $det\left(\mathbf{\Sigma_n}\left(\rho_1\right)\right)$ denotes the determinant of matrix $\mathbf{\Sigma_n}\left(\rho_1\right)$. Taking into account the Toeplitz symmetric structure of $\mathbf{\Sigma_n}\left(\rho_1\right)$, it can be shown that \cite{toeplitz}:
\begin{equation}
\mathbf{\Sigma_n}^{-1}\left(\rho_1\right) = \frac{1}{1-\rho_1^2}\left[\begin{array}{cccccc}
1&-\rho_1&0&\cdots&0&0\\
-\rho_1&1+\rho_1^2&-\rho_1&0&\cdots&0\\
0&-\rho_1&1+\rho^2_1&-\rho_1&0&\cdots\\
\vdots&\ddots&\ddots&\ddots&&\vdots\\
0&0&\cdots&0&-\rho_1&1
\end{array}\right].\label{eq-inv}
\end{equation}
By combining Eq.(\ref{jointpdf}) and (\ref{eq-inv}), the joint PDF $f_y\left(\mathbf{y}\right)$ is written as:
\begingroup
\begin{eqnarray}
f_y\left(\textbf{y}\right) &=& [det\left(\mathbf{\Sigma_n}\left(\rho_1\right)\right)]^{-1/2}\left(2\pi\right)^{-n/2}exp\left(-\frac{\left(y^2_1-2\rho_1 y_1y_2\right)}{2\left(1-\rho_1^2\right)}\right)\nonumber\\
&&\times exp\left(-\frac{ \sum^{n-1}_{i=2} \left(\left(\rho^2_1+1\right)y^2_i-2\rho_1 y_iy_{i+1}\right)}{2\left(1-\rho^2_1\right)}\right)\times exp\left(-\frac{ y^2_n}{2\left(1-\rho^2_1\right)}\right).\label{jointpdf2}
\end{eqnarray}
\endgroup
Therefore, the tridiagonal structure of $\mathbf{\Sigma_n}^{-1}\left(\rho_1\right)$ simplifies the theoretical estimation of the probability $Pr\{\mathfrak{b}_{A1_n}\}$, since the multiple integral in Eq.(\ref{eq:exp-int}) can be estimated by iteratively evaluating single integrals. More specifically, by substituting Eq.\eqref{jointpdf2} in Eq.\eqref{eq:exp-int}, the probability $Pr\{\mathfrak{b}_{A1_n}\}$ can be written as:
\begingroup
\begin{eqnarray}
&&Pr\{\mathfrak{b}_{A1_n}\} = C_0\int^{\infty}_{\frac{\gamma^*-\mu_{AR_n}}{\sigma_{AR_n}}}exp\left(-\frac{ y^2_n}{2\left(1-\rho_1^2\right)}\right)q_{n-1}\left(y_n\right)dy_n,\label{eq:exp-int2}
\end{eqnarray}
\endgroup
where $C_0 = [det\left(\mathbf{\Sigma_n}\left(\rho_1\right)\right)]^{-1/2}\left(2\pi\right)^{-n/2}$ and

\begingroup
\begin{eqnarray}
q_{k}\left(x\right) & = & \int^{\frac{\gamma^*-\mu_{AR_k}}{\sigma_{AR_k}}}_{-\infty}exp\left(-\frac{\left(\rho_1^2+1\right)y^2_k-2\rho_1 y_kx}{2\left(1-\rho_1^2\right)}\right)q_{k-1}(y_k)dy_k,\ k \in [2,n-1]\label{eq:qk}\\
q_1\left(x\right)   & = & \int^{\frac{\gamma^*-\mu_{AR_1}}{\sigma_{AR_1}}}_{-\infty}exp\left(-\frac{y^2_1-2\rho_1 y_1x}{2\left(1-\rho_1^2\right)}\right) dy_1.
\end{eqnarray}
\endgroup

\noindent In order to provide a closed-form expression for $q_1\left(x\right)$, we apply Eq.\cite[(15.74)]{book2}:

\begingroup
\begin{eqnarray}
\int^{\infty}_{0}exp\left(-\left(ax^2+bx+c\right)\right)dx &=&\sqrt{\frac{\pi}{a}}exp\left(\frac{b^2-4ac}{4a}\right)Q\left(\frac{b}{\sqrt{a}}\right).\label{closedform}
\end{eqnarray}
\endgroup
Hence, setting $t= \frac{\gamma^*-\mu_{AR_1}}{\sigma_{AR_1}}$, $a = 1/2(1-\rho_1^2)$, $b=(2\rho_1 x-2 t)/2(1-\rho_1^2)$ and $c=(t^2-2\rho_1 tx)/2(1-\rho_1^2)$, the integral $q_1\left(x\right)$ can be written as:

\begingroup

\begin{eqnarray}
q_1(x) & = & %\int^{t}_{-\infty}exp\left[-\frac{y^2_1-2\rho y_1x}{2\left(1-\rho^2\right)}\right] dy1\nonumber\\
%&=& \int^{\infty}_{0}exp\left[-\frac{\left(x-a\right)^2+2r\left(x-a\right)y_2}{2\left(1-r^2\right)}\right]dx\nonumber\\
%&=& \int^{\infty}_{0}exp\left[-\frac{\left(z^2+\left(2\rho x-2a\right)z+\left(t^2-2\rho tx\right)\right)}{2\left(1-\rho^2\right)}\right]dz\nonumber\\
%&=&
\sqrt{2\pi\left(1-\rho^2\right)} Q\left(\frac{\rho_1 x-t}{\sqrt{\left(1-\rho_1^2\right)}}\right)exp\left(\frac{\rho_1^2x^2}{2(1-\rho_1^2)}\right).
\end{eqnarray}
\endgroup

For the evaluation of the rest integrals $q_k\left(x\right)$, $\forall k \in [2,n-1]$, we adopt the Gaussian quadratures for the integral $\int^{\infty}_{0}exp\left(-x^2\right)f(x)dx$ \cite[Table II, N=15]{1969method}. Making changes in the variables\footnote{The detailed derivation is provided in Appendix \ref{a4}}, Eq.(\ref{eq:qk}) may be rewritten in the form:

\begingroup
\begin{eqnarray}
q_k(x) & = & \sqrt{\frac{2\left(1-\rho_1^2\right)}{1+\rho_1^2}}exp\left(-\frac{\left(\rho_1^2+1\right)t^2-2\rho_1 xt}{2\left(1-\rho_1^2\right)}\right) \sum^{N_{GQR}}_{i=1}w_i \nonumber\\
			 && \times exp\left(-\frac{\left(2\rho_1 x - 2t\left(\rho_1^2+1\right)\right)r_i}{\sqrt{2\left(1-\rho_1^4\right)}}\right)q_{k-1}\left(-\sqrt{\frac{2\left(1-\rho_1^2\right)}{1+\rho_1^2}}r_i+t\right)\label{eq:qk1},
\end{eqnarray}
\endgroup

\noindent where $t=\frac{\gamma^*-\mu_{AR_k}}{\sigma_{AR_k}}$, $w_i$ and $r_i$ denote the weights and the roots of the Gaussian quadratures\cite{1969method}, respectively, and $N_{GQR}$ is the number of points used for the integral evaluation. After evaluating $q_k(x)$ at the points $x_i = -\sqrt{\frac{2\left(1-\rho^2_1\right)}{1+\rho^2_1}}r_i+t$, $\forall k \in [2,n-1]$, the probability $Pr\{\mathfrak{b}_{A1_n}\}$ may be computed as:

\begingroup
\begin{eqnarray}
Pr\{\mathfrak{b}_{A1_n}\} & = & [det\left(\mathbf{\Sigma_n}(\rho_1)\right)]^{-1/2}\left(2\pi\right)^{-N/2}\sum^{N_{GQR}}_{i=1}w_iexp\left(-\frac{t^2 - 2\sqrt{2\left(1-\rho_1^2\right)}tr_i}{2\left(1-\rho_1^2\right)}\right)\nonumber\\
&& \times q_{N-1}\left(t-\sqrt{2\left(1-\rho_1^2\right)}r_i\right),
\end{eqnarray}
\endgroup

\noindent where $t=\frac{\gamma^*-\mu_{AR_n}}{\sigma_{AR_n}}$.

Following the same line of thought, the above procedure can be generalized for the theoretical estimation of $Pr\{\mathfrak{b}_{A\chi_n}\}$, $Pr\{\mathfrak{b}_{B\psi_n}\}$ $\forall$ $\chi,\psi\in \left[0,2^n-1\right]$, as it is described in Appendix \ref{a1}.

\subsubsection{Network Outage Probability ($p_{out}$)}

The network outage probability $p_{out}$, i.e., the probability that none of the relays in the system has successfully received both packets from nodes $A$ and $B$, may be directly derived from Eq.(\ref{P-k-n-2}) by setting $k=0$. Therefore:

\begingroup
\small
\begin{equation}
p_{out} = \mathop{\sum_{\chi=0}^{2^n-1}\sum_{\psi=0}^{2^n-1}}_{\{H_w\left(\mathfrak{b}_{A\chi_n}\odot \mathfrak{b}_{B\psi_n}\right) = 0\}}Pr\{\mathfrak{b}_{A\chi_n}\}Pr\{\mathfrak{b}_{B\psi_n}\}, \label{P-out}
\end{equation}
\endgroup

\noindent where $Pr\{\mathfrak{b}_{A\chi_n}\}$, $Pr\{\mathfrak{b}_{B\psi_n}\}$ are computed as described in Appendix \ref{a1}.

\vspace{-1pt}

\subsubsection{Expected Size of the Active Relay Set ($\mathbf{E}\left[\left|\mathcal{A}_n\right|\right]$)}
\label{subsec}

In this section, we provide a closed-form expression to compute the average number of active relays $\mathbf{E}\left[\left|\mathcal{A}_n\right|\right]$, proving that it is independent of the correlation coefficients $\rho_1, \rho_2$. Following the induction method, we initially prove the aforementioned statement for a network with 2 relays. By applying Eq.(\ref{P-k-n-2}) for $n=2$, we may write the probabilities that $k$ relays are active for $k=1,2$ as:

\vspace{-2pt}
\begingroup
{\setlength{\arraycolsep}{0em}\begin{eqnarray}
Pr\left\{\left|\mathcal{A}_2\right|=2\right\} &=& Pr\{\mathfrak{b}_{A3_2}\} Pr\{\mathfrak{b}_{B3_2}\} \label{p-2-2}\\
Pr\left\{\left|\mathcal{A}_2\right|=1\right\} &=& Pr\{\mathfrak{b}_{A1_2}\}\left(Pr\{\mathfrak{b}_{B1_2}\}+Pr\{\mathfrak{b}_{B3_2}\}\right)\nonumber\\
&&+Pr\{\mathfrak{b}_{A2_2}\}\left(Pr\{\mathfrak{b}_{B2_2}\}+ Pr\{\mathfrak{b}_{B3_2}\}\right)\nonumber\\
&&+Pr\{\mathfrak{b}_{A3_2}\}\left(Pr\{\mathfrak{b}_{B1_2}\}+Pr\{\mathfrak{b}_{B2_2}\}\right), \label{p-1-2}
\end{eqnarray}}
\endgroup
\vspace{-2pt}

\noindent where the probabilities $ Pr\{\mathfrak{b}_{A\chi_2}\}$, $Pr\{\mathfrak{b}_{B\psi_2}\}$, $\chi,\psi\in [0,3]$ may be written as in Eq.(\ref{eq:exp-int}) by setting $n=2$. The average number of active relays in case of $n=2$ may be written as follows:

\begingroup
\begin{eqnarray}
\label{eq:E-2}
\mathbf{E}\left[\left|\mathcal{A}_2\right|\right] &=& \sum^2_{i=1} iPr\left\{\left|\mathcal{A}_2\right|=i\right\}.
\end{eqnarray}
\endgroup

\vspace{-1pt}

By taking into account the Eq.(\ref{eq:exp-int}), (\ref{closedform}) and Lemma 1 below, it can be shown\footnote{The detailed derivation is provided in Appendix \ref{a5}} that Eq.(\ref{eq:E-2}) can be written in closed-form as:
\vspace{-1pt}
\begingroup
\begin{eqnarray}
\mathbf{E}\left[\left|\mathcal{A}_2\right|\right] &=& \sum^2_{i=1} Q\left(\frac{\gamma^* - \mu_{ARi}}{\sigma_{ARi}}\right)Q\left(\frac{\gamma^* - \mu_{{BRi}}}{\sigma_{{BRi}}}\right)\label{eq:average2}.
\end{eqnarray}
\endgroup

\begin{lemma}
For any given $\rho$, $\gamma^*$, $\mu$, $\sigma$ it holds that:

\begingroup
\begin{eqnarray}
\frac{1}{\sqrt{2\pi}}\int^{\infty}_{-\infty}Q\left(\frac{\gamma^* - \mu-\sigma \sqrt{\rho} t}{\sigma\sqrt{1-\rho}}\right)e^{-t^2/2} dt
=Q\left(\frac{\gamma^* - \mu}{\sigma}\right)
\end{eqnarray}
\endgroup
\end{lemma}
\begin{proof}
The proof of Lemma 1 is given in Appendix \ref{a2}.
\end{proof}

\noindent The generalization of this result for a network with $n$ relays may be stated as follows:

\begin{proposition}
Let us consider the cooperative network of Fig. \ref{f1} operating in a correlated shadowing environment. The average number of active relays $E\left[\left|\mathcal{A}_n\right|\right]$ is independent of the correlation between the links and it is given by:

\begingroup
\begin{equation}
\mathbf{E}\left[\left|\mathcal{A}_n\right|\right] {=} \sum^n_{i=1} Q\left(\frac{\gamma^* - \mu_{ARi}}{\sigma_{ARi}}\right)Q\left(\frac{\gamma^* - \mu_{{BRi}}}{\sigma_{{BRi}}}\right).
\label{eq:Average}
\end{equation}
\endgroup
\end{proposition}
\begin{proof}
The proof of Proposition 1 is given in Appendix \ref{a3}.
\end{proof}

Having derived accurate closed-form expressions for crucial network parameters (i.e., the network outage probability and the expected number of active relays), we can now incorporate them into a MAC layer analytical model to study important end-to-end metrics, such as the network throughput and energy efficiency.

\subsection{Analytical Formulation from the MAC Layer Perspective}
\label{sec:analysis_mac}

\subsubsection{Throughput}
\label{sec:throughput}

The network throughput, measured in b/s, is defined as the rate of successful data delivery in a given period of time. Hence, taking into account the protocol operation, the expected total throughput ($S_{total}$) of the network, can be decomposed into the throughput achieved by the direct successful transmissions ($S_d$) and the throughput produced by the relay nodes ($S_{coop}$) during the cooperation phase. This can be mathematically expressed as:

\begin{equation}
\label{eq:Stotal}
 	\mathbf{E}[S_{total}]=\mathbf{E}[S_{d}]+\mathbf{E}[S_{coop}],
\end{equation}

where

\begin{equation}
\label{eq:SD}
 	\mathbf{E}[S_{d}]=(1-OPER_{AB})\cdot \frac{\mathbf{E}[Payload]}{\mathbf{E}[T_d]}
\end{equation}

and

\begin{equation}
\label{eq:Scoop}
 	\mathbf{E}[S_{coop}]=2\cdot OPER_{AB}\cdot (1-p_{out})\cdot \frac{\mathbf{E}[Payload]}{\mathbf{E}[T_{d}]+\mathbf{E}[T_{coop}]}.
\end{equation}

In the above equations, $OPER_{AB}$ corresponds to the $OPER$ in the direct link from node $A$ to $B$, $\mathbf{E}[Payload]$ is the average packet payload and $p_{out}$ denotes the probability that there are no active relays in the network, i.e., all relays are in outage. In addition, $\mathbf{E}[T_d]$ and $\mathbf{E}[T_{coop}]$ represent the average time required for a successful direct transmission and a transmission that takes place via the relays, respectively. Let us also emphasize that NC techniques enable the simultaneous transmission of two data packets and, hence, the coefficient 2 has to be included in Eq.(\ref{eq:Scoop}).

The average time for the direct transmission ($\mathbf{E}[T_d]$) can be estimated by the total packet size (including MAC and PHY headers) and the transmission data rate ($Data\ Tx.Rate$) as:
\begin{equation}
\label{eq:Td}
 	\mathbf{E}[T_d]= \frac{\mathbf{E}[Packet\ Size]}{Data\ Tx.Rate}.
\end{equation}

On the other hand, the term $\mathbf{E}[T_{coop}]$ can be written as the sum of the minimum deterministic default time ($T_{def}$) in the beginning of the cooperation, and the overhead time due to the contention of the relays:

\begin{equation}
\label{eq:Tcoop}
 	\mathbf{E}[T_{coop}]= T_{def}+\mathbf{E}[T_{ovh}].
\end{equation}

The default time, which mainly corresponds to the transmission for the RFC and the data packet b, is equal to:

\begin{equation}
\label{eq:Tdef}
T_{def}=T_{SIFS}+T_{RFC}+T_{b},
\end{equation}
where $T_{RFC}$ and $T_{b}$ denote the transmission time for RFC and data packet $b$, respectively, while $T_{SIFS}$ corresponds to the SIFS duration. On the other hand, the overhead time can be caused due to either the network outage or the contention phase:

\begin{equation}
\label{eq:Tovh}
\mathbf{E}[T_{ovh}]=p_{out}\cdot T_{timeout}+(1-p_{out})\cdot\mathbf{E}[T_{cont}],
\end{equation}
where $p_{out}$ is the network outage probability, $T_{timeout}$ denotes the period of time that all nodes wait in case of no active relay in the network, and $T_{cont}$ represents the total time duration until the correct acknowledgement of both original packets, equal to:

\begin{equation}
\label{eq:Tcont}
 	\mathbf{E}[T_{cont}]=T_{ONC}+T_{DIFS}+\mathbf{E}[T_{C}]+T_{a \oplus b}+2\cdot T_{SIFS}+2\cdot T_{ACK}.
\end{equation}

Eq.(\ref{eq:Tcont}) explicitly considers: i) the expected time required for a coded packet to be transmitted via the relays ($\mathbf{E}[T_{C}]$), taking into account the idle slots and the collision overhead, ii) the overhead time needed to perform NC ($T_{ONC}$), iii) the sensing times $T_{DIFS}$ and $T_{SIFS}$, iv) the transmission time for the NC packet $T_{a \oplus b}$, and v) the transmission time for the ACK packets ($T_{ACK}$).

Since NCCARQ is characterized by backwards compatibility with the IEEE 802.11 Standard, the channel access can be modeled according to the Markov chain introduced in \cite{bianchi}, where the states correspond to the values of the backoff counter and the transition probabilities follow the DCF operation. Hereafter, we provide the slightly modified formulas for the sake of the paper's self-completeness, while the interested reader should be referred to the Appendix of \cite{nccarq} for the detailed protocol analysis. Thus, the average time until a successful transmission is calculated as:

\begin{equation}
\label{eq:Tc}
\mathbf{E}[T_{C}]=(\frac{1}{p_s}-1)[(\frac{p_i}{1-p_s})T_{slot}+(\frac{p_c}{1-p_s})T_{col}],
\end{equation}
where $T_{slot}$ represents the idle slot duration and $T_{col}$ corresponds to the collision time, equal to: $T_{col}=T_{DIFS}+T_{a \oplus b}+T_{SIFS}$. In addition, the probabilities of having an idle ($p_i$), a successful ($p_s$), or a collided ($p_c$) slot can be written as:

\begin{equation}
\label{eq:pi}
 	p_i=1-p_{tr}
\end{equation}
\begin{equation}
\label{eq:ps}
 	p_s=p_{tr}\cdot p_{s|tr}
\end{equation}
\begin{equation}
\label{eq:pc}
 	p_c=p_{tr}\cdot(1-p_{s|tr}),
\end{equation}
where $p_{tr}$ is the probability that at least one relay attempts to transmit:
\begin{equation}
\label{eq:ptr}
 	p_{tr}=1-(1-\tau)^{\mathbf{E}[|\mathcal{A}_n|]}
\end{equation}
and $p_{s|tr}$ denotes the probability of a successful transmission (i.e., exactly one station transmits conditioned on the fact that at least one station transmits):
\begin{equation}
\label{eq:pstr}
 	p_{s|tr}=\frac {\mathbf{E}[|\mathcal{A}_n|]\tau (1-\tau)^{\mathbf{E}[|\mathcal{A}_n|]-1}}{1-(1-\tau)^{\mathbf{E}[|\mathcal{A}_n|]}}.
\end{equation}

In Eq.(\ref{eq:ptr}) and (\ref{eq:pstr}), $\tau$ is the probability that a node transmits in a randomly selected slot and $\mathbf{E}[|\mathcal{A}_n|]$ is the expected number of active relays during the cooperation phase as we have seen in Section~\ref{subsec}. It is worth noting that traditional MAC-oriented analytical works usually neglect the impact of the PHY layer by including the total number of relays~($n$) in the theoretical expressions, while in our work, this set is restricted by taking into account realistic PHY layer conditions.

\subsubsection{Energy Efficiency}
\label{sec:energy}

The network energy efficiency, measured in b/J, can be defined as the amount of transmitted useful information per energy unit. Considering the protocol operation, the expected energy efficiency~($\eta$) may be written as:

\begin{equation}
\label{eq:Eef}
 	\mathbf{E}[\eta]=\frac{(1-OPER_{AB})\cdot\mathbf{E}[Payload]+2\cdot OPER_{AB}\cdot(1-p_{out})\cdot\mathbf{E}[Payload]}{\mathbf{E}[\mathcal{E}_{total}]},
\end{equation}
where the numerator corresponds to the expected number of delivered useful bits during one communication round, and the denominator represents the average energy consumption at the same time period.

Regarding the expected total energy consumption in the network, following the same line of thought, we discompose the operation into the direct transmission and the cooperation phase. Hence:

\begin{equation}
\label{eq:Etotal}
 	\mathbf{E}[\mathcal{E}_{total}]=\mathcal{E}_{d}+OPER_{AB}\cdot \mathbf{E}[\mathcal{E}_{coop}].
\end{equation}

Let us recall that the network consists of two nodes ($A$ and $B$) that exchange packets with the assistance of $n$ relays. Defining as $P_{Tx}$, $P_{Rx}$ and $P_{idle}$ the power levels associated to the transmission ($Tx$), reception ($Rx$) and idle mode, respectively, the energy consumption during the direct transmissions can be estimated as:

\begin{equation}
\label{eq:Ed}
 	\mathcal{E}_{d}=P_{Tx}\cdot T_a+(n+1)\cdot P_{Rx}\cdot T_a.
\end{equation}

On the other hand, the term $\mathbf{E}[\mathcal{E}_{coop}]$ is composed of the energy consumption during the network outage and the energy consumed in the successful cooperation:

\begin{equation}
\label{eq:Ecoop}
 	\mathbf{E}[\mathcal{E}_{coop}]=p_{out}\cdot \mathcal{E}_{out}+(1-p_{out})\cdot \mathbf{E}[\mathcal{E}_{suc\_coop}],
\end{equation}
where the $\mathbf{E}[\mathcal{E}_{suc\_coop}]$ includes the required energy for a perfectly scheduled cooperative phase ($\mathcal{E}_{min}$), and the energy consumption during the contention phase ($\mathbf{E}[\mathcal{E}_{cont}]$):

\begin{equation}
\label{eq:Ecoop}
 	\mathbf{E}[\mathcal{E}_{suc\_coop}]=\mathcal{E}_{min}+\mathbf{E}[\mathcal{E}_{cont}].
\end{equation}

Hence, considering the network topology and the protocol's operation, we have:

\begin{equation}
\label{eq:Eout}
 	\mathcal{E}_{out}=(n+2)\cdot P_{idle}\cdot T_{timeout}
\end{equation}

\footnotesize
\begin{equation}
\label{eq:Emin}
\begin{aligned}
&\mathcal{E}_{min}=(n+2)\cdot P_{idle}\cdot T_{SIFS}+P_{Tx}\cdot (T_{RFC}+T_B)+(n+1)\cdot P_{Rx}\cdot(T_{RFC}+T_B)+(n+2)\cdot P_{idle}\cdot T_{ONC}+\\
&+(n+2)\cdot P_{idle}\cdot T_{DIFS}+P_{Tx}\cdot T_{a\oplus b}+2\cdot P_{Rx} \cdot T_{a\oplus b}+(n-1)\cdot P_{idle}\cdot T_{a\oplus b}+(n+2)\cdot P_{idle}\cdot T_{SIFS}+\\
&+2\cdot P_{Tx}\cdot T_{ACK}+2\cdot (n+1)\cdot P_{Rx}\cdot T_{ACK}+(n+2)\cdot P_{idle} \cdot T_{SIFS}\\
\end {aligned}
\end{equation}
\normalsize

\begin{equation}
\label{eq:Econt}
 	\mathbf{E}[\mathcal{E}_{cont}]=p_i\cdot ((n+2)\cdot P_{idle}\cdot T_{slot})+p_c\cdot (\mathbf{E}[L]\cdot P_{Tx}\cdot T_{col}+2\cdot P_{Rx}\cdot T_{col}+(n-\mathbf{E}[L])\cdot P_{idle}\cdot T_{col}),
\end{equation}
where $\mathbf{E}[L]$ represents the average number of relays that transmit a packet simultaneously. Given the existence of $\mathbf{E}[|\mathcal{A}_n|]$ active relays, the probability $p_l$ that exactly $l$ stations are involved in a collision can be expressed as:

\begin{equation}
\label{eq:pk}
 	p_l=\frac{\dbinom{\mathbf{E}[|\mathcal{A}_n|]}{l}\tau^{l}(1-\tau)^{\mathbf{E}[|\mathcal{A}_n|]-l}}{p_c}
\end{equation}
and, therefore:

\begin{equation}
\label{eq:Ek}
 	\mathbf{E}[L]=\displaystyle\sum_{l=2}^{\mathbf{E}[|\mathcal{A}_n|]}l\cdot p_{l}=\displaystyle\sum_{l=2}^{\mathbf{E}[|\mathcal{A}_n|]}l\cdot\frac{\dbinom{\mathbf{E}[|\mathcal{A}_n|]}{l}\tau^{l}(1-\tau)^{\mathbf{E}[|\mathcal{A}_n|]-l}}{p_c}.
\end{equation}

\section{Model Validation and Performance Evaluation}
\label{sec:performance}

We have developed a MATLAB simulator that incorporates both the NCCARQ rules and the PHY layer design, in order to validate our analytical model and study the impact of exponentially correlated shadowing on the performance of NC-based MAC protocols. In this section, we present the simulation setup along with the results of our experiments.

\subsection{Simulation Setup}
\label{sec:setup}

The considered network, depicted in Fig. \ref{f1}, consists of two nodes ($A$ and $B$) that participate in a bidirectional wireless communication, and $n$ relay nodes that contribute to the data exchange. In the same figure, the shadowing correlation between the different links is highlighted, assuming that: i) all $AR_i$ links are exponentially correlated as described in Eq.(\ref{eq1}), ii) all $BR_i$ links are also exponentially correlated according to Eq.(\ref{eq2}), and iii) pairs of $AR_i$ and $BR_i$ links are independent, which is a reasonable assumption according to measurements in \cite{cor1}. Furthermore, we adopt a symmetric network topology with $\rho_1=\rho_2=\rho$.

The MAC layer parameters have been selected in line with the IEEE 802.11g Standard specifications \cite{80211std}. In particular, the initial Contention Window ($CW$) for all nodes is 32, the MAC header overhead is 34~bytes, while the time for the application of NC to the data packets is considered negligible, as the coding takes place only between two packets. We also consider time slots, SIFS, DIFS and timeout interval of 20, 10, 50 and 80~$\mu$s, respectively. In addition, based on the work of Ebert et al. \cite{ebert} on the power consumption of the wireless interface, we have chosen the following power levels for our scenarios: $P_{Tx}=1900$~mW and $P_{Rx}=P_{idle}=1340$~mW.

Regarding the PHY layer parameters, we have set the reliability threshold $\gamma^*=41.14=16.14$~dB, which corresponds to a target $APER=10^{-1}$. Furthermore, we assume a relatively weak direct ($AB$) link ($\mu_{AB}=8$~dB) with respect to the SNR threshold $\gamma^*$, in order to trigger the cooperation and focus our study on the impact of correlated shadowing. The simulation parameters are summarized in Table \ref{t2}. Through the experimental assessment, we want to validate our proposed models and study the effect of the number of relays ($n$) and the correlation factor ($\rho$) on the protocol performance.

\begin{table}[htb]
\caption{System Parameters} \label{t2}
\begin{center}
\begin{tabular}{|c|c||c|c|}
\hline
\textbf{Parameter} & \textbf{Value} & \textbf{Parameter} & \textbf{Value} \\ \hline
\textit{Packet Payload} & 1500 bytes & \textit{$CW_{min}$} & 32\\ \hline
\textit{$T_{slot}$} & 20 $\mu$s & \textit{$T_{timeout}$} & 80 $\mu$s \\ \hline
\textit{SIFS} & 10 $\mu$s & \textit{DIFS} & 50 $\mu$s\\ \hline
\textit{MAC Header} & 34 bytes & \textit{PHY Header} & 96 $\mu$s \\ \hline
\textit{Data Tx.Rate} & 54 Mb/s & \textit{Control Tx.Rate} & 6 Mb/s \\ \hline
$\gamma^*$ & 16.14 dB & $\sigma$ & [0,10] dB\\ \hline
$\mu_{AR_i}=\mu_{BR_i}$ & \{15,20\} dB & $\mu_{AB}$ & 8 dB\\ \hline
\textit{$P_{Tx}$} & 1900 mW & \textit{$P_{Rx}$} & 1340 mW \\ \hline
\textit{$P_{idle}$} & 1340 mW & $\rho$ & [0,1) \\ \hline
 \end{tabular}
\end{center}
\end{table}

\subsection{Model Validation}
\label{sec:validation}

In the first set of our experiments, we study the PHY layer impact on the communication, while we validate the derived theoretical expressions. Fig. \ref{phyf1}-\ref{val2} depict the expected size of the active relay set ($\mathbf{E}[|\mathcal{A}_n|]$), the network outage probability ($p_{out}$), the expected network throughput ($\mathbf{E}[S_{total}]$) and the expected energy efficiency ($\mathbf{E}[\eta]$), respectively, for different values of the shadowing standard deviation $\sigma$, assuming strong links between the end nodes~($A$,$B$) and the relays~($R_i$), i.e., $\mu_{AR_i}=\mu_{BR_i}=20$~dB.

\begin{figure}[htb]
\centering
\includegraphics[width=1\columnwidth]{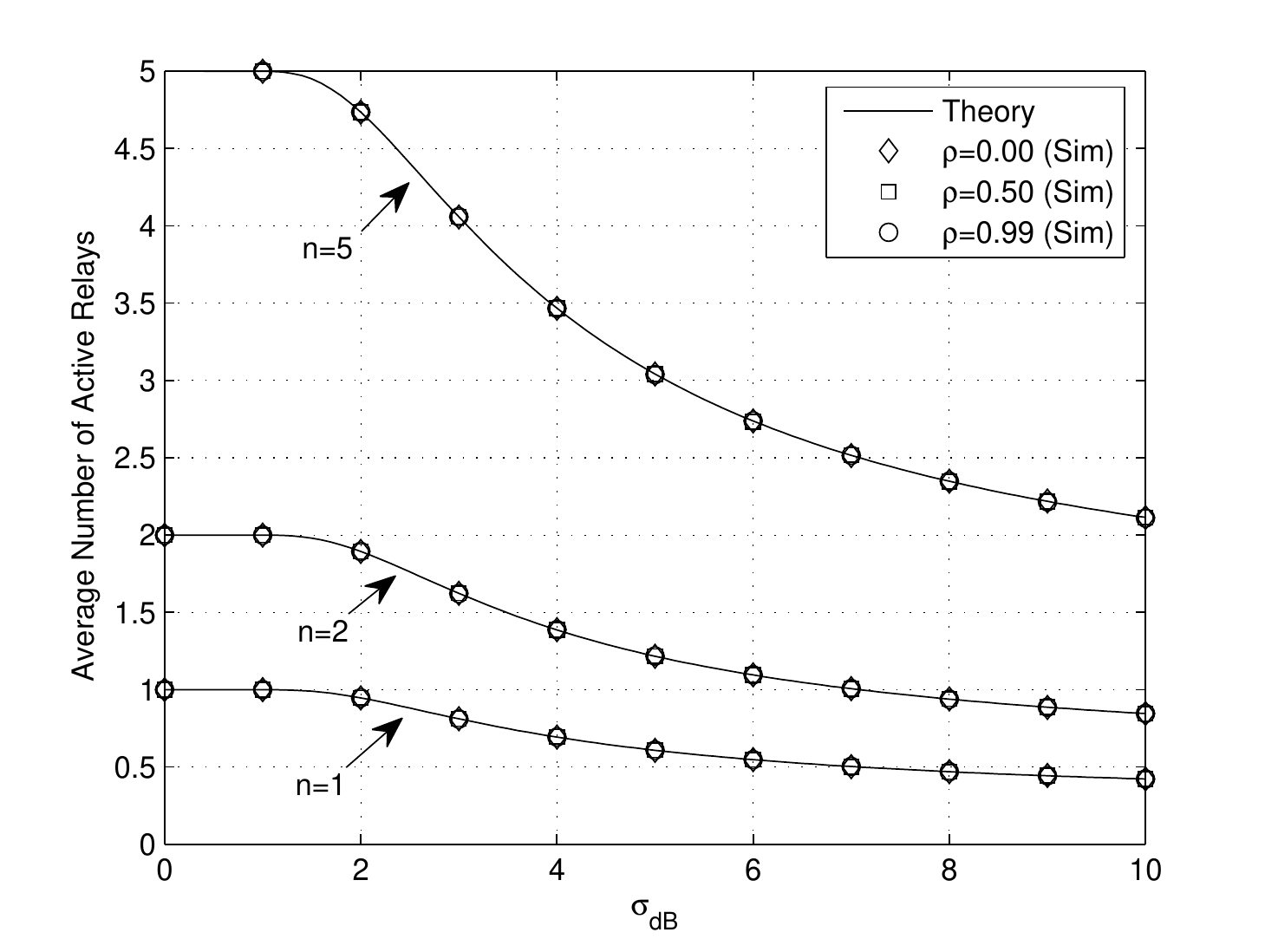}
\caption{Expected size of the active relay set ($\mathbf{E}[|\mathcal{A}_n|]$) vs. Shadowing standard deviation ($\sigma$) ($\mu_{AR_i}=\mu_{BR_i}=20$~dB)}\label{phyf1}
\end{figure}

In Fig. \ref{phyf1}, we consider different total number of relays and various indicative values for the correlation factor ($\rho$), deriving two important conclusions. First, the experiments validate our analysis, demonstrating that the average number of active relays is independent of the shadowing correlation among the wireless links. The second important remark concerns the negative effect of $\sigma$ in the number of active relays. In this particular scenario, where the mean SNR value is above the threshold $\gamma^*$, the shadowing variation has a detrimental role in the communication. As a result, higher values of $\sigma$ restrict the potential diversity benefits by reducing the expected size of the active relay set. However, it should be mentioned that in the opposite case (i.e., when the mean SNR value is below the reliability threshold), high values of $\sigma$ would imply lower outage probability and higher number of active relays, hence increasing the expected network throughput in the system, as we will examine in the following section.

Fig. \ref{phyf2} illustrates the theoretical and simulation results for the network outage probability for different correlation factors ($\rho$) and number of relays ($n$). Similar to the previous case, the shadowing deviation deteriorates the system performance, increasing the probability of having no active relay in the system. However, in this case, the impact of shadowing correlation on the system is clearly demonstrated in the figure, since high values of $\rho$ cause almost identical outage probability for the network independently of $n$, annulling the advantages of the distributed cooperation. On the other hand, independent wireless links ($\rho=0$) exploit the diversity offered by the relays, considerably reducing the outage probability as the total number of relays in the system increases (e.g., $n=5$). As a result, the significant effect of $\rho$ on the probability of outage has a direct impact on the end-to-end metrics under study, highlighting the importance of having the exact knowledge of the shadowing correlation conditions in the network.

\begin{figure}[htb]
\centering
\includegraphics[width=1\columnwidth]{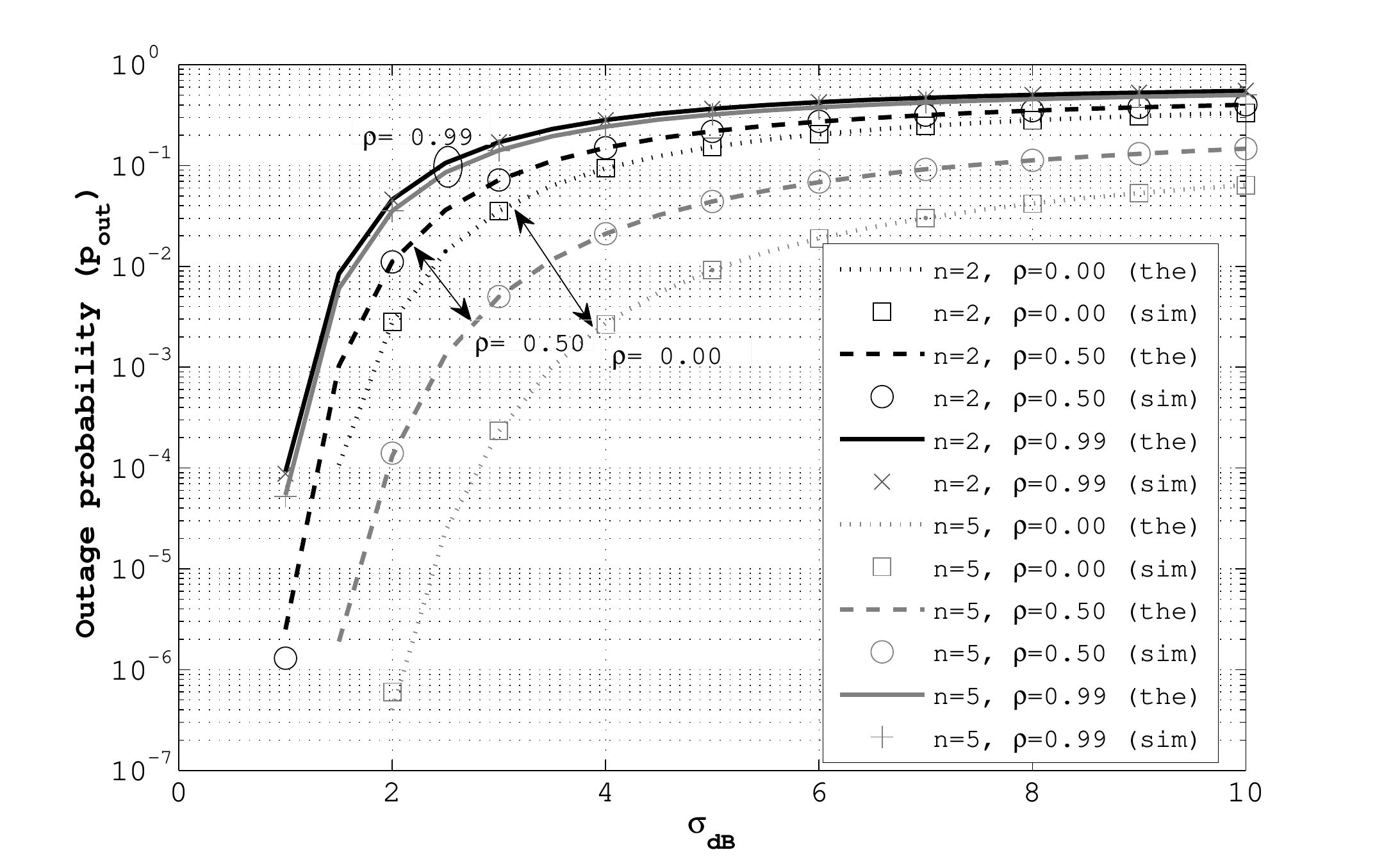}
\caption{Network outage probability ($p_{out}$) vs. Shadowing standard deviation ($\sigma$) ($\mu_{AR_i}=\mu_{BR_i}=20$~dB)}\label{phyf2}
\end{figure}

In Fig. \ref{val1}, we study the impact of shadowing standard deviation ($\sigma$) on the network throughput for different number of relay nodes ($n$). In this specific case, where $\mu_{AR_i}=\mu_{BR_i}>\gamma^*$, the wireless communication would always be successful without the shadowing random fluctuations and, hence, shadowing is harmful for the system, as it introduces many events where the received SNR is below the threshold $\gamma^*$. In addition, two important remarks are highlighted: i) distributed cooperation is beneficial, as the throughput increases with the number of available relays ($n$), and ii) shadowing correlation is detrimental to the potential gain introduced by cooperation.

\begin{figure}[htb]
\centering
\includegraphics[width=1\columnwidth]{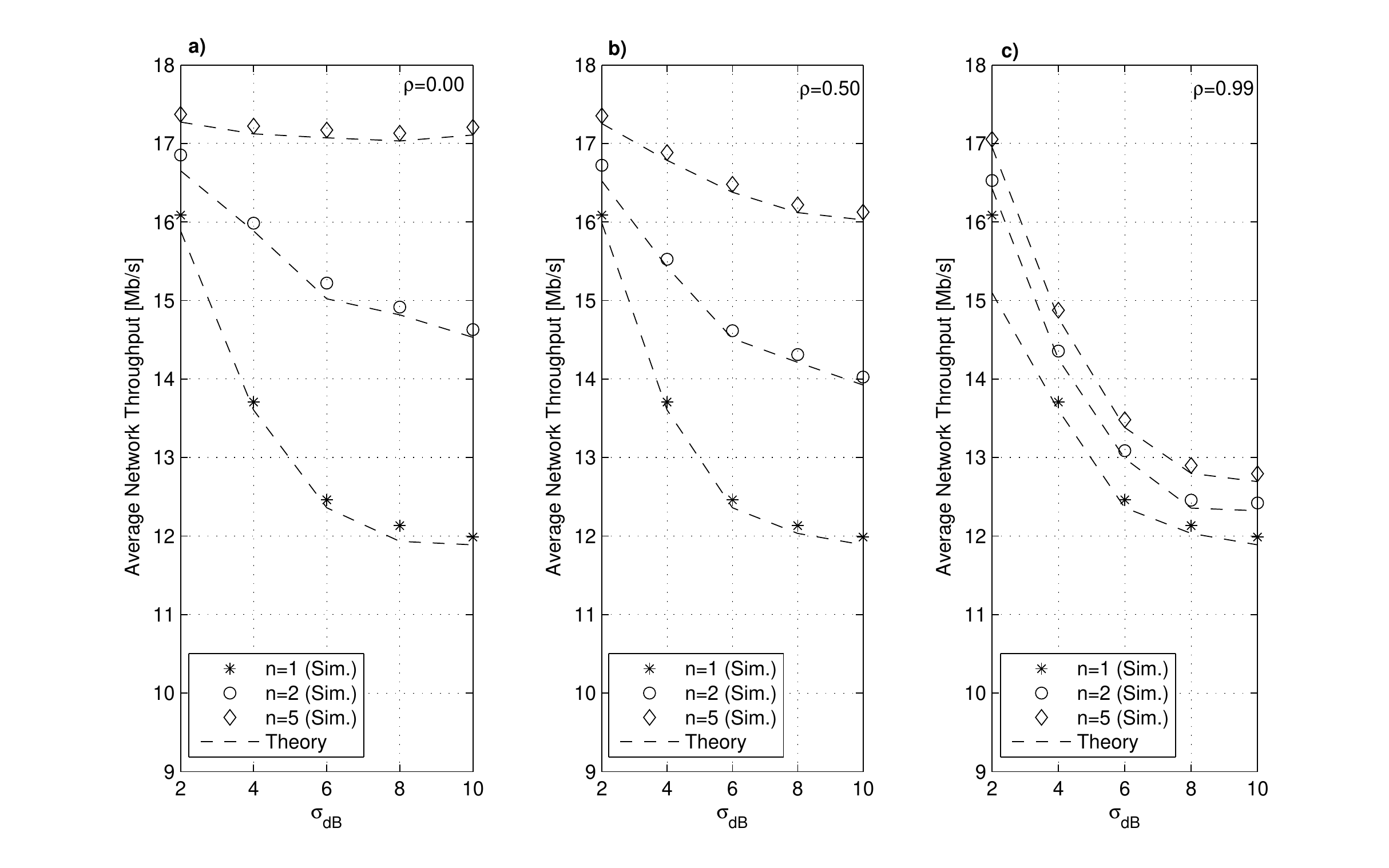}
\caption{Average network throughput ($\mathbf{E}[S_{total}]$) vs. shadowing standard deviation ($\sigma$), assuming $\mu_{AR_i}=\mu_{BR_i}=20$~dB and a) $\rho=0.00$, b) $\rho=0.50$, c) $\rho=0.99$} \label{val1}
\end{figure}

The expected network energy efficiency for different number of relays and correlated conditions is plotted in Fig. \ref{val2}, validating our model and revealing intriguing facets of the problem, since they disclose a notable trade-off between the system throughput and energy efficiency. In particular, although distributed cooperation provides significant gains in the throughput for high SNR scenarios (Fig. \ref{val1}), it has a negative impact on the energy efficiency, reducing it up to 100\% under specific conditions. This fact can be explained by taking into account the high throughput (12~Mb/s) achieved in single-relay networks under good channel conditions. Cooperation may increase this performance up to 18~Mb/s, but the aggregated energy consumption of many relays in the network results in a significant reduction of the total energy efficiency. The next section presents a thorough performance evaluation with regard to the impact of the number of relays and the shadowing correlation on the network performance.

\begin{figure}[htb]
\centering
\includegraphics[width=1\columnwidth]{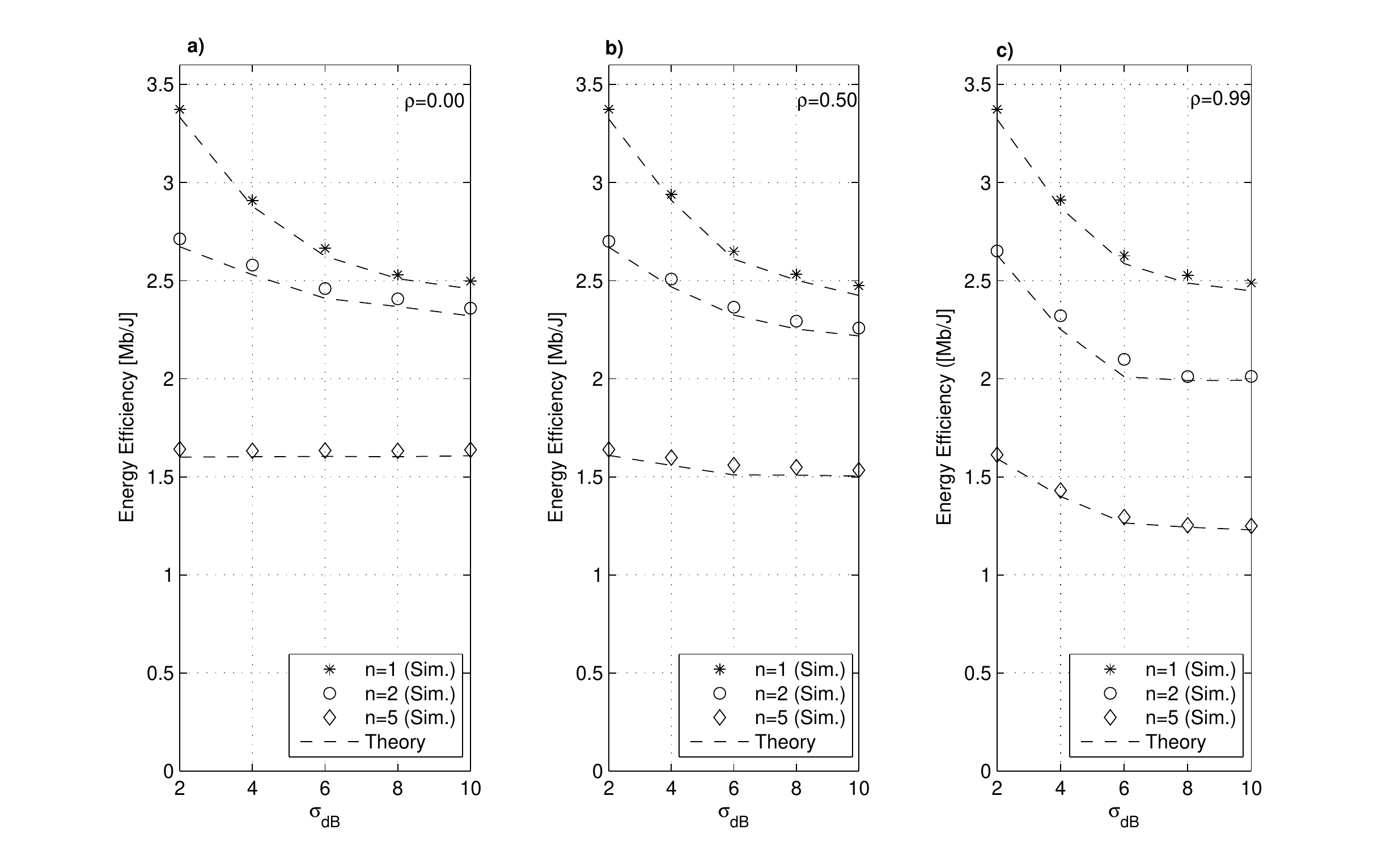}
\caption{Average energy efficiency ($\mathbf{E}[\eta]$)vs. shadowing standard deviation ($\sigma$), assuming $\mu_{AR_i}=\mu_{BR_i}=20$~dB and a) $\rho=0.00$, b) $\rho=0.50$, c) $\rho=0.99$} \label{val2}
\end{figure}

\subsection{Performance Evaluation}

In Fig. \ref{th1}-\ref{ee2}, we study the impact of the correlation and the number of relays ($n$) on the network throughput and energy efficiency, for three different topologies (i.e., $\rho=0$, $\rho=0.5$ and $\rho=0.99$). In this set of experiments, we have set $\mu_{AR_i}=\mu_{BR_i}=15$~dB, which is a value close to the reliability threshold ($\gamma^*$). In addition, in order to emphasize the importance of the shadowing standard deviation, we have adopted two extreme values of $\sigma$, i.e., $\sigma=2$ dB and $\sigma=10$ dB.

\subsubsection{Impact of the number of relays ($n$) in the network}

Fig. \ref{th1} studies the network throughput versus the number of relays, deriving three important conclusions. First, the high number of relays in the network is beneficial for the average system throughput, especially for low and medium values of $\rho$, where the incorporation of more relays in the network results in higher diversity. The second interesting remark is related to the influence of the shadowing correlation factor on the protocol performance. In particular, as expected by studying the outage probability, high values of $\rho$ compensate the benefits from the cooperation, as the throughput increases only slightly with the number of relays. However, for medium and low values of the $\rho$, the protocol performance seems to remain unaffected, which implies that the network throughput is not proportional to the correlation among the links. This observation would be particularly important for network design, since it allows the network deployment under relatively high correlation conditions (e.g. $\rho=0.5$), although they may sound prohibitive on principle. Finally, comparing Fig. \ref{th1}a and Fig. \ref{th1}b, we observe the impact of $\sigma$ on the throughput, as high values of the shadowing standard deviation ($\sigma$=10 dB) significantly increase the protocol performance, especially for small number of relays in the network. For example, in the case of single-relay systems ($n=1$), the throughput is almost quadruple, something that can be intuitively explained by taking into account that the mean value of the SNR ($\mu_{AR_i}=\mu_{BR_i}$) is marginally lower than the decoding threshold and, as a result, the random fluctuations introduced by shadowing due to high values of $\sigma$ enable the correct packet decoding more often.

\begin{figure}[htb]
\centering
\includegraphics[width=1\columnwidth]{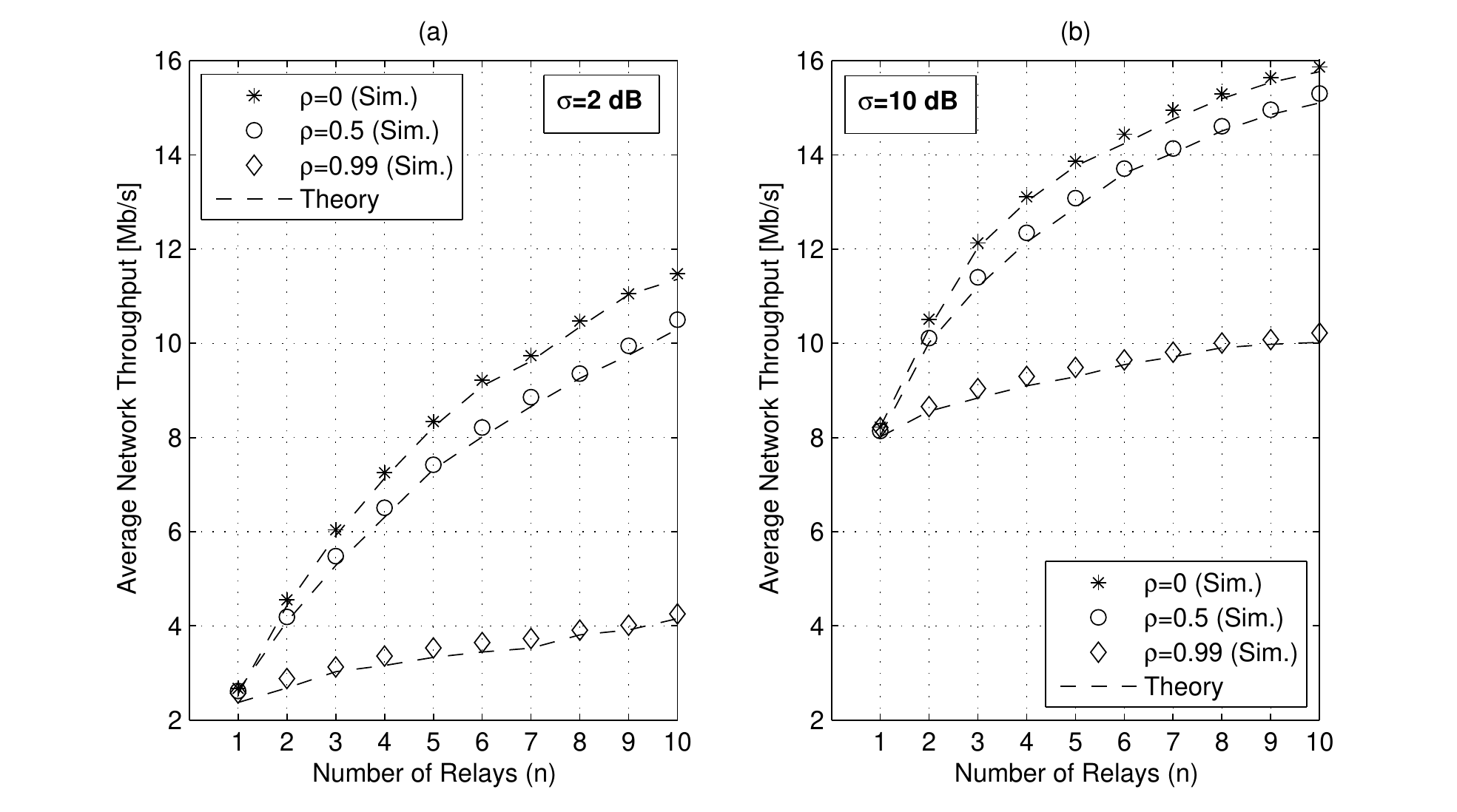}
\caption{Average network throughput ($\mathbf{E}[S_{total}]$) vs. Number of relays ($n$) for $\rho=0$, $\rho=0.5$, $\rho=0.99$, considering: a) $\sigma$=2~dB and b) $\sigma$=10~dB ($\mu_{AR_i}=\mu_{BR_i}=15$~dB)} \label{th1}
\end{figure}

Figures \ref{ee1}a and \ref{ee1}b present the network energy efficiency for $\sigma$=2~dB and $\sigma$=10~dB, respectively. The first clear outcome from both figures is the negative role of shadowing correlation in the energy efficiency of the network. In addition, the shadowing standard variation plays again an important role, as the energy efficiency increases with $\sigma$, mainly due to the significant throughput increase in the network. However, as the number of relays increases, the relative gain due to $\sigma$ decreases, since the respective throughput gains are significantly higher for low values of $\sigma$. Moreover, it is worth commenting on the different behavior of the plots in each figure. In Fig. \ref{ee1}a, we observe that, for low and medium correlation factors, the energy efficiency increases until $n=4$, as the throughput gains we achieve by adding relays in the network deserve the increased energy consumption in the system. For higher number of relays (i.e., $n>4$), the energy efficiency decreases slowly, since the incorporation of more relays (which need extra energy resources) in the network does not fully justify the increase in throughput. In high correlated scenarios (i.e., $\rho$=0.99), as expected, the energy efficiency decreases as the number of relays ggrwos, since, due to the almost identical conditions in the wireless links, the throughput of the network is not significantly affected. Regarding the case of $\sigma$=10~dB (Fig. \ref{ee1}b), we can see that the increase in the number of relays causes a significant reduction in the network energy efficiency, although the average throughput (Fig. \ref{th1}b) follows again an increasing trend. However, in this case, we should consider that the throughput with only one relay in the network is relatively high and, therefore, the price (in terms of energy) we have to pay by adding more relays is higher than the actual gains we get in terms of performance.

\begin{figure}[htb]
\centering
\includegraphics[width=1\columnwidth]{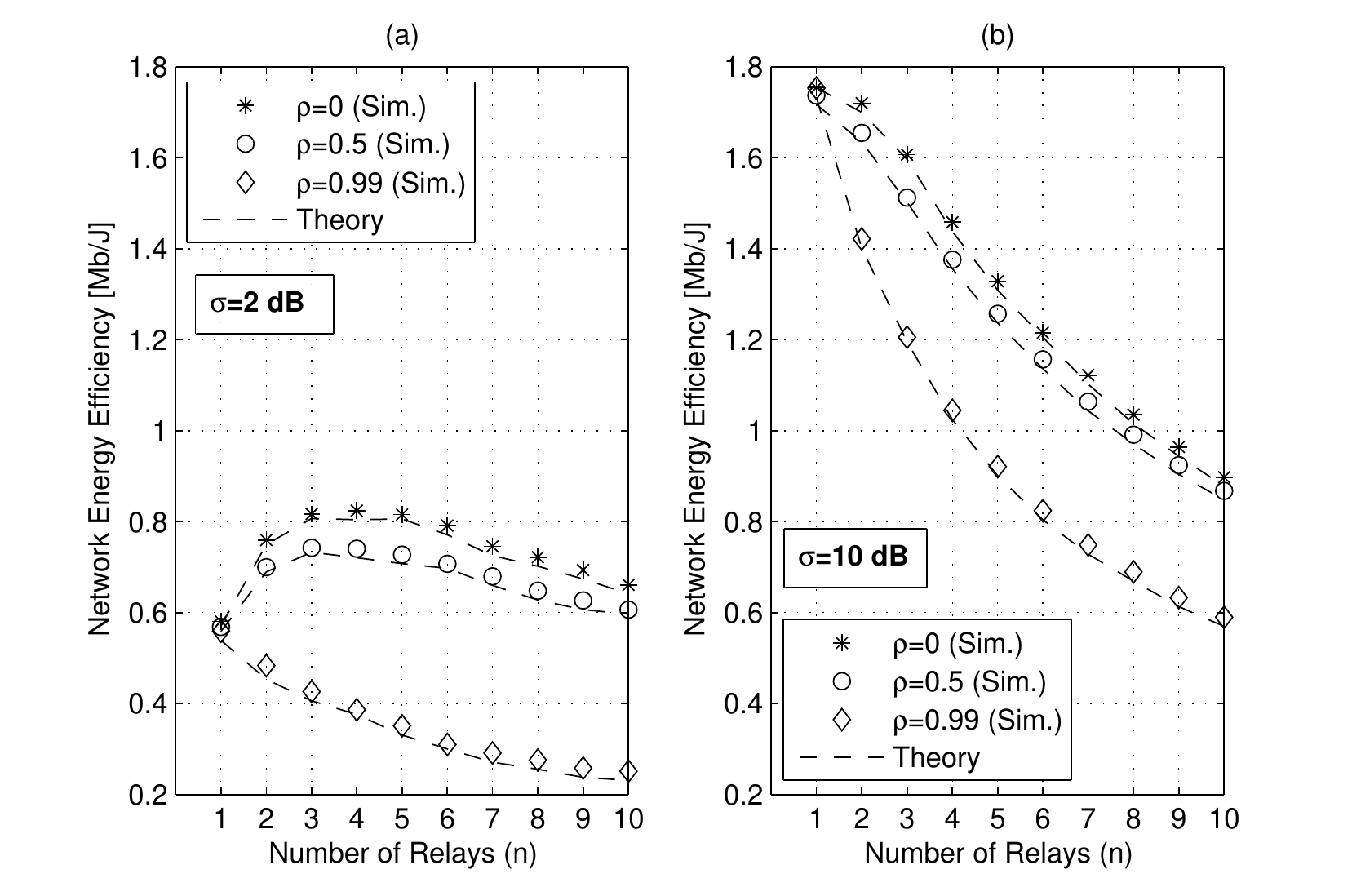}
\caption{Network energy efficiency ($\eta$) vs. Number of relays ($n$) for $\rho=0$, $\rho=0.5$, $\rho=0.99$, considering: a) $\sigma$=2~dB and b) $\sigma$=10~dB ($\mu_{AR_i}=\mu_{BR_i}=15$~dB)}\label{ee1}
\end{figure}

\subsubsection{Impact of the correlation factor ($\rho$)}

In Fig. \ref{th2}a and \ref{th2}b, we study the impact of the shadowing correlation factor ($\rho$) on the network throughput for $\sigma=2$ dB and $\sigma=10$ dB, respectively, while we also plot the case of one relay ($n=1$) in the network as a reference scenario, although the correlation in this case has no practical meaning. Once more, we confirm that the high number of relays, as well as high values of $\sigma$ are beneficial for the throughput in cooperative scenarios with the specific setup. With regard to the impact of shadowing correlation, in both cases, we observe that the cooperation tends to be useless for high correlation factors ($\rho\rightarrow 1$), since all the relays experience very similar shadowing attenuations, and the throughput reduces to that of a single-relay network. However, it can be remarked that the impact of correlation is more severe in environments with low $\sigma$, as the difference in the throughput performance in case of independent ($\rho=0$) and fully correlated ($\rho=0.99$) links is much higher in Fig. \ref{th2}a. In addition, we can verify the conclusions of the previous set of experiments, where it was shown that the results for $\rho=0$ and $\rho=0.5$ (which are common values for outdoor environments \cite{out_cor}) were not significantly different. In this figure, we can explicitly specify that the severe performance degradation occurs for extremely high values of the correlation factor (i.e., $\rho>0.7$), which are usually found in indoor environments \cite{in_cor1,in_cor2}.

\begin{figure}[htb]
\centering
\includegraphics[width=1\columnwidth]{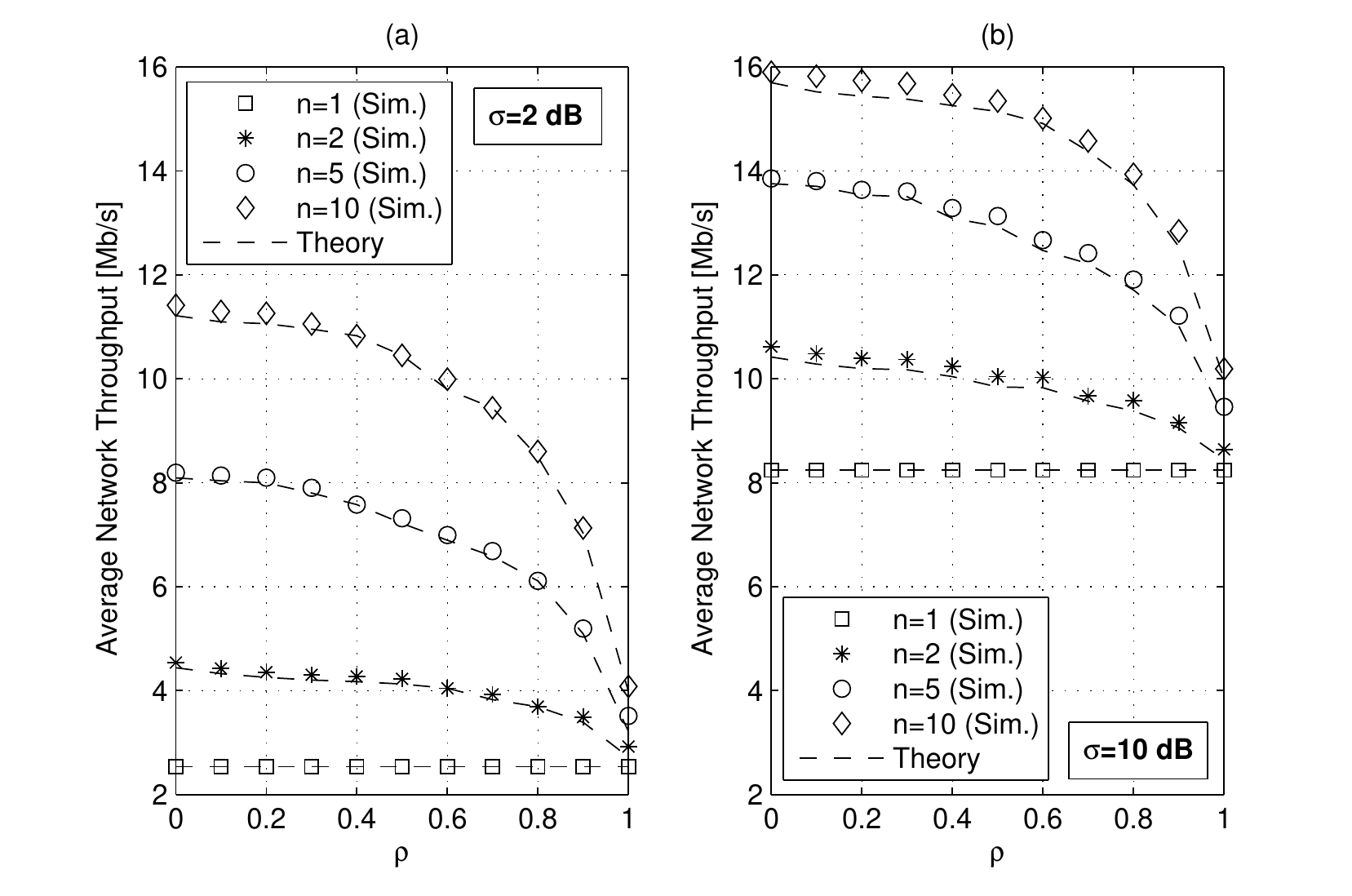}
\caption{Average network throughput ($\mathbf{E}[S_{total}]$) vs. Shadowing correlation factor ($\rho$) for $n=1,2,5,10$, considering: a) $\sigma$=2~dB and b) $\sigma$=10~dB ($\mu_{AR_i}=\mu_{BR_i}=15$~dB)}\label{th2}
\end{figure}

The impact of correlation on the network energy efficiency is shown in Fig.~\ref{ee2}. Starting from the case of $\sigma=10$ dB (Fig.~\ref{ee2}b), we can see that adding more relays in the network causes a considerable reduction in the energy efficiency, while the effect of correlation is not particularly harmful. This fact can be explained by comparing the throughput performance of networks with $n=1$ and $n=10$ relays in Fig. \ref{th2}b. Apparently, we need 10 relays in order to double the throughput of single-relay networks. However, the significantly increased energy consumption in the network is not in accordance with the throughput enhancement, thus resulting in lower energy efficiency in the network. The same conclusion can be also supported by noticing that, in contrast to the throughput performance results, the baseline network energy efficiency (i.e., with one relay in the network) is higher in all cases, since the protocol is able to achieve high performance standards under these conditions, even with only one relay in the system. On the other hand, the results in Fig.~\ref{ee2}a are not straightforward, as they identify the necessity of carefully choosing the number of relays in order to achieve the highest energy efficiency. Unlike Fig.~\ref{ee2}b, where the incorporation of additional relays in the system always results in energy efficiency degradation, in this case, the existence of more than one relay in the system, besides throughput, can also be beneficial for the energy efficiency of the network, especially in low correlation scenarios. For instance, in our particular experiment, we can see that the energy efficiency increases by adding few relays in the system (e.g., up to $n=5$), since the achieved throughput raises considerably even with only a small number of deployed relay nodes (Fig.~\ref{th2}a). As the number of relays increases (e.g., $n=10$), the energy consumption grows in higher rates than the throughput, which results in lower energy efficiency in the network. However, as the correlation among the wireless links increases (i.e., $\rho>0.7$), the deployment of multiple relays does not provide significant performance enhancement, something that is directly reflected in the energy efficiency, which drops significantly.

\begin{figure}[htb]
\centering
\includegraphics[width=1\columnwidth]{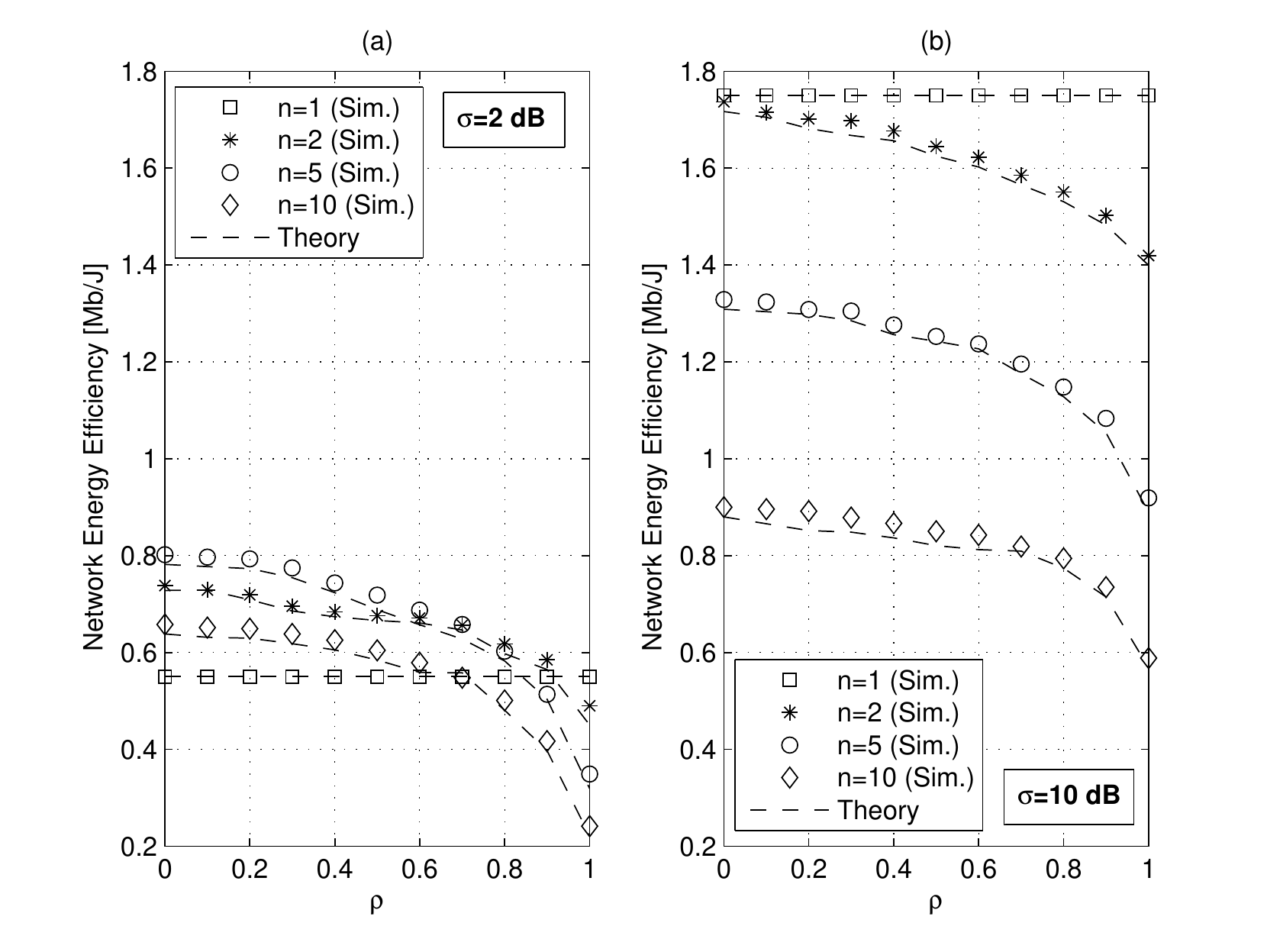}
\caption{Network energy efficiency ($\eta$) vs. Shadowing correlation factor ($\rho$) for $n=1,2,5,10$, considering: a) $\sigma$=2~dB and b) $\sigma$=10~dB ($\mu_{AR_i}=\mu_{BR_i}=15$~dB)}\label{ee2}
\end{figure}

In all cases, the experimental results clearly showcase that: i) the system performance is notably affected only by extremely high values of $\rho$, and ii) although shadowing correlation does not have an influence on the average number of active relays, it significantly affects the outage probability and, thus, the network should be designed taking into account the exact physical parameters and application requirements. In addition, very interesting tradeoffs between the throughput and the energy efficiency in the network have been revealed by the extensive performance assessment. In particular, the throughput improvement offered by the distributed cooperation comes with the respective energy costs that should not be neglected. The incorporation of many relays in the communication increases the total energy consumption in the network, without yielding the expected throughput gains in all cases. More specifically, in highly correlated scenarios where the cluster of relays does not offer significant throughput gains, the energy efficiency is remarkably reduced. It is also worth noting that shadowing variations can be either beneficial or harmful for the communication, depending on the quality of the cooperative links in the network. To that end, the proposed cross-layer analytical model provides the network designer with accurate estimations that facilitate the decision for the optimum number of relays in the network and their best possible placement in order to reduce the deployment and operational cost, guaranteeing, at the same time, the desired network throughput.

\section{Concluding Remarks}
\label{sec:conclusions}

In this paper, we have proposed a cross-layer analytical framework to model end-to-end metrics (i.e., throughput and energy efficiency), in two-way cooperative networks under realistic correlated shadowing conditions. The proposed model jointly considers the MAC layer operation along with crucial PHY layer parameters, such as the network outage probability and the average number of active relays in the network. The extensive performance assessment has revealed interesting tradeoffs between throughput and energy efficiency, while the PHY layer analysis has demonstrated that the average number of relays is independent of the shadowing correlation in the wireless links. The proposed analytical model can provide useful insights that can be exploited for effective network planning in realistic channel conditions. In our future work, we plan to study the temporal shadowing correlation and design effective cross-layer mechanisms that enhance the performance of the state-of-the-art NC-aided MAC protocols.

\appendices % for no appendix heading

\section{Derivation of Equation (\ref{eq:qk1})}
\label{a4}

After setting $y_k = -z+t$, where $t = \frac{\gamma^* -\mu_{AR_k}}{\sigma_{AR_k}}$, Eq. (\ref{eq:qk}) may be written as:

\begingroup
\begin{eqnarray}
q_k(x) & = & \int^{t}_{-\infty}exp\left[-\frac{\left(\rho_1^2+1\right)y^2_k-2\rho_1 y_kx}{2\left(1-\rho_1^2\right)}\right]q_{k-1}(y_k)dy_k\nonumber\\
			 & = & \int^{\infty}_{0}exp\left[-\frac{\left(\rho_1^2+1\right)\left(z-t\right)^2+2\rho_1 x\left(z-t\right)}{2\left(1-\rho_1^2\right)}\right]q_{k-1}(-z+t)dz\nonumber\\
			 %& = & exp\left[-\frac{\left(r^2+1\right)a^2-2rxa}{2\left(1-r^2\right)}\right] \int^{\infty}_{0}exp\left[-\frac{\left(r^2+1\right)}{2\left(1-r^2\right)}z^2\right]\nonumber\\
			 %&& \times exp\left[-\frac{2rxz - 2a\left(r^2+1\right)z}{2\left(1-r^2\right)}\right]q2(-z+a)dz\nonumber\\
			 & = & \sqrt{\frac{2\left(1-\rho_1^2\right)}{1+\rho_1^2}}exp\left[-\frac{\left(\rho_1^2+1\right)t^2-2\rho_1 xt}{2\left(1-\rho_1^2\right)}\right] \int^{\infty}_{0}exp\left[-z^2\right]\nonumber\\
			 && \times exp\left[-\frac{\left(2\rho_1 x - 2t\left(\rho_1^2+1\right)\right)z}{\sqrt{2\left(1-\rho_1^4\right)}}\right]q_{k-1}(-\sqrt{\frac{2\left(1-\rho_1^2\right)}{1+\rho_1^2}}z+t)dz.\label{eq:qk2}
\end{eqnarray}
\endgroup

\noindent The aforementioned expression may be evaluated using the Gaussian quadratures for the integral $\int^{\infty}_{0}exp\left[-x^2\right]f(x)dx$, originally proposed in \cite[Table II, N=15]{1969method}, as follows:
\begin{equation}
\int^{\infty}_{0}exp\left[-x^2\right]f(x)dx \approx \sum^{N_{GQR}}_{i=1}w_i f(r_i),
\end{equation}

\noindent where $w_i$, $r_i$ denote the weights and roots of the Gaussian quadratures as defined in \cite[Table II, N=15]{1969method}, and $N_{GQR}$ is the number of points used for the integral evaluation. By using the aforementioned formula, we may write Eq. \eqref{eq:qk2} as follows:
\begingroup
\begin{eqnarray}
q_k(x) & = & \sqrt{\frac{2\left(1-\rho_1^2\right)}{1+\rho_1^2}}exp\left[-\frac{\left(\rho_1^2+1\right)t^2-2\rho_1 xt}{2\left(1-\rho_1^2\right)}\right] \sum^{N_{GQR}}_{i=1}w_i \nonumber\\
			 && \times exp\left[-\frac{\left(2\rho_1 x - 2t\left(\rho_1^2+1\right)\right)r_i}{\sqrt{2\left(1-\rho_1^4\right)}}\right]q_{k-1}(-\sqrt{\frac{2\left(1-\rho_1^2\right)}{1+\rho_1^2}}r_i+t),
\end{eqnarray}
\endgroup
thus deriving Eq. \eqref{eq:qk1}.

\section{Iterative Computation of $Pr\{\mathfrak{b}_{A\chi_n}\}$ by Using the Gaussian Quadratures}
\label{a1}

\begin{flushleft}
\texttt{Step 1: Evaluate $q_1\left(x\right)$ at the points $x_i = \sqrt{\frac{2\left(1-\rho_1^2\right)}{1+\rho_1^2}}r_i-t$ if
$\boldsymbol{1}_{AR_1}$ or $x_i = -\sqrt{\frac{2\left(1-\rho_1^2\right)}{1+\rho_1^2}}r_i+t$ if $\boldsymbol{0}_{AR_1}$, according to:}
\end{flushleft}

\begingroup
\footnotesize
\begin{equation}
q_1(x) = \left\{\begin{array}{c}\sqrt{2\pi\left(1-\rho_1^2\right)} Q\left(\frac{t-\rho_1 x}{\sqrt{\left(1-\rho_1^2\right)}}\right)exp\left(\frac{\rho_1^2x^2}{2(1-\rho_1^2)}\right) \ if\ \boldsymbol{1}_{AR_1}\\
\sqrt{2\pi\left(1-\rho_1^2\right)} Q\left(\frac{\rho_1 x-t}{\sqrt{\left(1-\rho_1^2\right)}}\right)exp\left(\frac{\rho_1^2x^2}{2(1-\rho_1^2)}\right)\ otherwise\end{array}\right.,\  t = \frac{\gamma^*-\mu_{AR_1}}{\sigma_{AR_1}}\label{eq-1}
\end{equation}
\endgroup

\begin{flushleft}
\texttt{Step k: Evaluate $q_k\left(x\right)$, $\forall k \in[2,n-1]$ at the points, $x_i = \sqrt{\frac{2\left(1-\rho_1^2\right)}{1+\rho_1^2}}r_i-t$ if
$\boldsymbol{1}_{AR_k}$ or $x_i = -\sqrt{\frac{2\left(1-\rho_1^2\right)}{1+\rho_1^2}}r_i+t$ if $\boldsymbol{0}_{AR_k}$ according to:}
\end{flushleft}

\begingroup
\tiny
\begin{eqnarray}
q_k(x) & = & \left\{
\begin{array}{c}
\sqrt{\frac{2\left(1-\rho_1^2\right)}{1+\rho_1^2}}exp\left(-\frac{\left(\rho_1^2+1\right)t^2-2\rho_1 xt}{2\left(1-\rho_1^2\right)}\right)\sum^{N_{GQR}}_{i=1}w_i exp\left(\frac{\left(2\rho_1 x - 2t\left(\rho_1^2+1\right)\right)r_i}{\sqrt{2\left(1-\rho_1^4\right)}}\right)q_{k-1}\left(\sqrt{\frac{2\left(1-\rho_1^2\right)}{1+\rho_1^2}}r_i-t\right)\ if \ \boldsymbol{1}_{AR_k}\\
 \sqrt{\frac{2\left(1-\rho_1^2\right)}{1+\rho_1^2}}exp\left(-\frac{\left(\rho_1^2+1\right)t^2-2\rho_1 xt}{2\left(1-\rho_1^2\right)}\right)\sum^{N_{GQR}}_{i=1}w_i exp\left(-\frac{\left(2\rho_1 x - 2t\left(\rho_1^2+1\right)\right)r_i}{\sqrt{2\left(1-\rho_1^4\right)}}\right)q_{k-1}\left(-\sqrt{\frac{2\left(1-\rho_1^2\right)}{1+\rho_1^2}}r_i+t\right)\ if \ \boldsymbol{0}_{AR_k}
\end{array}\right.,\
\end{eqnarray}
\endgroup

\begin{flushleft}
\texttt{where  $t = \frac{\gamma^*-\mu_{AR_k}}{\sigma_{AR_k}}$}
\end{flushleft}

\begin{flushleft}
\texttt{Step n: Evaluate $Pr\{\mathfrak{b}_{A\chi_n}\}$ according to:}
\end{flushleft}

\begingroup
\footnotesize
\begin{eqnarray}
Pr\{\mathfrak{b}_{A\chi_n}\} & = & \left\{
\begin{array}{c}
C_0\sum^{N_{GQR}}_{i=1}w_iexp\left(-\frac{t^2 + 2\sqrt{2\left(1-\rho_1^2\right)}tr_i}{2\left(1-\rho_1^2\right)}\right)q_{n-1}\left(\sqrt{2\left(1-\rho_1^2\right)}r_i+t\right)\ if\ \boldsymbol{1}_{AR_n}\\
C_0\sum^{N_{GQR}}_{i=1}w_iexp\left(-\frac{t^2 - 2\sqrt{2\left(1-\rho_1^2\right)}tr_i}{2\left(1-\rho_1^2\right)}\right)q_{n-1}\left(t-\sqrt{2\left(1-\rho_1^2\right)}r_i\right)\ if\ \boldsymbol{0}_{AR_n}
\end{array}\right.
\end{eqnarray}
\endgroup

\begin{flushleft}
\texttt{where $C_0 = [det\left(\mathbf{\Sigma}\left(\rho_1\right)\right)]^{-1/2}\left(2\pi\right)^{-n/2}$, $t = \frac{\gamma^*-\mu_{AR_n}}{\sigma_{AR_n}}$.}
\end{flushleft}

\section{Proof of Lemma 1}
\label{a2}

Any RV $Z\sim N(\mu,\sigma^2)$ may be written in the form:
\vspace{-1pt}
\begingroup
\begin{equation}
Z=  \sigma \sqrt{\rho} X_1 + \sigma \sqrt{1-\rho} X_2 + \mu
\end{equation}
\endgroup
%\vspace{-1pt}
for any given $0\leq \rho< 1$, where $X_1,X_2\sim N(0,1)$. Considering $W Z-\mu$, $ X =\sigma \sqrt{\rho} X_1$, $Y=\sigma \sqrt{1-\rho} X_2$, the above equation may be written as $W=X+Y$,
%\vspace{-1pt}
%\begingroup
%\footnotesize
%\begin{equation}
%W=  X + Y,
%\end{equation}
%\endgroup
%\vspace{-1pt}
\noindent where $W\sim N\left(0,\sigma^2\right)$, $X\sim N \left(0,\sigma^2\rho \right)$, $Y \sim N\left(0,\sigma^2 \left(1-\rho\right)\right)$.
%
%The pdf of W can be written as:
%\vspace{-1pt}
%\begingroup
%\footnotesize
%\begin{eqnarray}
%f_w(w) &=&  \int^{+\infty}_{-\infty}f_y(w-x) f_x(x) dx\nonumber
%\end{eqnarray}
%\endgroup
%\vspace{-1pt}
The Cumulative Distribution Function (CDF) of $W$ can be written as:
\vspace{-1pt}
\begingroup
\begin{eqnarray}
F_w(t) %&=&  \int^{t}_{-\infty} \int^{+\infty}_{-\infty}f_y(w-x) f_x(x) dxdw = \nonumber\\%\int^{+\infty}_{-\infty}\int^{t}_{-\infty}f_y\left(w-x\right)dw f_x(x) dx\nonumber\\
&=& \int^{+\infty}_{-\infty}F_y\left(t-x\right)f_x(x) dx \label{eq:57}
\end{eqnarray}
\endgroup
%\vspace{-1pt}
After making some changes in variables, Eq.(\ref{eq:57}) may be written as follows:
\vspace{-1em}

\begingroup
\begin{eqnarray}
%F_z(t+\mu) &=& \int^{+\infty}_{-\infty}F_y\left(t-x\right)f_x(x) dx\nonumber\\
%F_z(t) &=& \int^{+\infty}_{-\infty}F_y\left(t-x-\mu\right)f_x(x) dx\nonumber\\
F_z(t) &=& \frac{1}{\sqrt{2\pi \sigma \rho}}\int^{+\infty}_{-\infty}F_{x_2}\left(\frac{t-x-\mu}{ \sigma \sqrt{1-\rho} }\right)e^{-\frac{x^2}{2\sigma^2 \rho}} dx
\end{eqnarray}
\endgroup
and, consequently:
\begingroup
\begin{eqnarray}
%\left(1-Q\left(\frac{t-mu}{ \sigma }\right)\right)&=&\frac{1}{\sqrt{2\pi \sigma \rho}}\int^{+\infty}_{-\infty}\left(1-Q\left(\frac{t-x-mu}{ \sigma \sqrt{1-\rho^2} }\right)\right)e^{-\frac{x^2}{2\sigma^2 \rho^2}} dx\nonumber\\
Q\left(\frac{t-\mu}{ \sigma }\right) &=&\frac{1}{\sqrt{2\pi}}\int^{+\infty}_{-\infty}Q\left(\frac{t-\sigma\sqrt{\rho}x_1-\mu}{ \sigma \sqrt{1-\rho} }\right)e^{-\frac{x^2}{2}}dx.
%&=& \frac{1}{\sqrt{2\pi}}\int^{+\infty}_{-\infty}}F_{x_2}\left(\frac{t-x-mu}{ \sigma \sqrt{1-\rho^2} }\right)e^{-\frac{x_1^2}} dx_1
\end{eqnarray}
\endgroup

\noindent The above equation concludes our proof.

\section{Derivation of Equation (\ref{eq:average2})}
\label{a5}

In the special case where the number of available relays in the system is equal to 2, then, alternatively to Eq. (3), the two correlated links $\bar{\gamma}_{AR_{1}}, \bar{\gamma}_{AR_{2}}$ can be generated as follows:

\begingroup
{\setlength{\arraycolsep}{0em}
\begin{eqnarray}
\bar{\gamma}_{AR_{1}} &=& \sigma_{AR_{1}} \left(\sqrt{1-\rho_1^2}X_1+\rho_1 X_0\right)+\mu_{AR_{1}}\label{Rvs}\\
\bar{\gamma}_{AR_{2}} &=& \sigma_{AR_{2}} \left(\sqrt{1-\rho_1^2}X_2+\rho_1 X_0\right)+\mu_{AR_{2}}\label{Rvs2},
\end{eqnarray}}
\endgroup

\noindent with $X_i\sim \mathcal{N}(0,1)$, $\forall i$. In that case, we may write:

\begingroup
{\setlength{\arraycolsep}{0em}
\begin{eqnarray}
&&Pr\{\mathfrak{b}_{A1_2}\} = Pr\left\{\boldsymbol{0}_{AR_1},\boldsymbol{1}_{AR_2}\right\}\nonumber\\
&& = \int^{\infty}_{-\infty}\Pr\left\{\bar{\gamma}_{AR_{1}}\leq \gamma^*,\bar{\gamma}_{AR_{2}}> \gamma^*\left|X_0=t\right.\right\}f_{X_0}\left(t\right)dt\nonumber\\
&& = \int^{\infty}_{-\infty}\Pr\left\{X_1\leq a_1(t)\right\}\Pr\left\{X_2> a_2(t)\right\}f_{X_0}\left(t\right)dt\nonumber\\
&& = \frac{1}{\sqrt{2\pi}}\int^{\infty}_{-\infty}\left(1-Q\left(a_{1}\left(t\right)\right)\right) Q\left(a_{2}\left(t\right)\right)e^{\frac{-t^2}{2}}dt, \label{A1}
\end{eqnarray}}
\endgroup

\noindent where $a_{i}(t) = \left(\gamma^*- \mu_{AR_{i}}-\sigma_{AR_{i}} \rho_1 t\right)/\left(\sigma_{AR_{i}}\sqrt{1-\rho^2_1}\right)$ and $f_{X_0}\left(t\right)$ is the probability density function of $X_0$. Similarly, we may write:

\begingroup
{\setlength{\arraycolsep}{0em}
\begin{eqnarray}
&&Pr\{\mathfrak{b}_{A2_2}\} = \frac{1}{\sqrt{2\pi}}\int^{\infty}_{-\infty}Q\left(a_{1}\left(t\right)\right)\left(1 - Q\left(a_{2}\left(t\right)\right)\right)e^{\frac{-t^2}{2}}dt \label{A2}\\
&&Pr\{\mathfrak{b}_{A3_2}\} = \frac{1}{\sqrt{2\pi}}\int^{\infty}_{-\infty}Q\left(a_{1}\left(t\right)\right)Q\left(a_{2}\left(t\right)\right)e^{\frac{-t^2}{2}}dt. \label{A3}
\end{eqnarray}}
\endgroup

\noindent Equivalently, the probabilities $Pr\{\mathfrak{b}_{B\psi_2}\}$ can be written as follows:

\begingroup
{\setlength{\arraycolsep}{0em}
\begin{eqnarray}
&&Pr\{\mathfrak{b}_{B1_2}\} = \frac{1}{\sqrt{2\pi}}\int^{\infty}_{-\infty}\left(1-Q\left(b_{1}\left(t\right)\right)\right) Q\left(b_{2}\left(t\right)\right)e^{\frac{-t^2}{2}}dt \label{B1}\\
&&Pr\{\mathfrak{b}_{B2_2}\} = \frac{1}{\sqrt{2\pi}}\int^{\infty}_{-\infty}Q\left(b_{1}\left(t\right)\right)\left(1 - Q\left(b_{2}\left(t\right)\right)\right)e^{\frac{-t^2}{2}}dt \label{B2}\\
&&Pr\{\mathfrak{b}_{B3_2}\} = \frac{1}{\sqrt{2\pi}}\int^{\infty}_{-\infty}Q\left(b_{1}\left(t\right)\right)Q\left(b_{2}\left(t\right)\right)e^{\frac{-t^2}{2}}dt, \label{B3}
\end{eqnarray}}
\endgroup

\noindent where $b_{i}(t) = \left(\gamma^*- \mu_{BR_{i}}-\sigma_{BR_{i}} \rho_2 t\right)/\left(\sigma_{BR_{i}}\sqrt{1-\rho^2_2}\right)$. By combining Eq. (24) - (26), the average number of active relays is given by:

\begingroup
\begin{eqnarray}
\label{eq:E-3}
\mathbf{E}\left[\left|\mathcal{A}_2\right|\right] &=& \sum^2_{i=1} iPr\left\{\left|\mathcal{A}_2\right|=i\right\}\nonumber\\
&=& \left(Pr\{\mathfrak{b}_{A3_2}\} + Pr\{\mathfrak{b}_{A1_2}\}\right)\left(Pr\{\mathfrak{b}_{B3_2}\}+Pr\{\mathfrak{b}_{B1_2}\}\right) \nonumber\\
&&+ \left(Pr\{\mathfrak{b}_{A3_2}\} + Pr\{\mathfrak{b}_{A2_2}\}\right)\left(Pr\{\mathfrak{b}_{B3_2}+Pr\{\mathfrak{b}_{B2_2}\}\right). \label{average}
\end{eqnarray}
\endgroup

By substituting Eq. (\ref{A1}) - (\ref{B3}) in Eq.  (\ref{average}), we get:

\begingroup
\begin{eqnarray}
\label{eq:E-2}
\mathbf{E}\left[\left|\mathcal{A}_2\right|\right] &=& \sum^2_{i=1} \frac{1}{\sqrt{2\pi}}\int^{\infty}_{-\infty} Q\left(a_{i}\left(t\right)\right)e^{\frac{-t^2}{2}}dt\nonumber\\
&& \times\frac{1}{\sqrt{2\pi}}\int^{\infty}_{-\infty} Q\left(b_{i}\left(t\right)\right)e^{\frac{-t^2}{2}}dt.
\end{eqnarray}
\endgroup

Finally, applying Lemma 1 in Eq. (\ref{eq:E-2}), we derive Eq. (27).

\section{Proof of Proposition 1}
\label{a3}

In Section \ref{subsec}, it has been proven that the average number of active relays in a bidirectional cooperative network with two relays is given by Eq.(\ref{eq:average2}). Let us assume that, for a network with $k$ relays, it holds that:

\begin{equation}
\mathbf{E}\left[\left|\mathcal{A}_k\right|\right]=\sum^k_{i=1}Q\left(\left( \gamma^* - \mu_{AR_i}\right)/\sigma_{AR_i}\right)Q\left(\left( \gamma^* - \mu_{BR_i}\right)/\sigma_{BR_i}\right).
\end{equation}

Using the induction method, we are going to prove that:

\begin{eqnarray}
\mathbf{E}\left[\left|\mathcal{A}_{k+1}\right|\right] =  \mathbf{E}\left[\left|\mathcal{A}_{k}\right|\right]+Q\left(\frac{\gamma^* - \mu_{AR_{k+1}}}{\sigma_{AR_{k+1}}}\right)Q\left(\frac{\gamma^* - \mu_{BR_{k+1}}}{\sigma_{BR_{k+1}}}\right).
\end{eqnarray}

Let us denote by $p^{k+1}_{R_i} = Pr\left\{\left|\mathcal{A}_{k+1}\right|=i\right\}$ the probability that $i$ relays are active in a network with $k+1$ relays, and by $B_{k+1} = \left\{\boldsymbol{1}_{AR_{k+1}},\boldsymbol{1}_{BR_{k+1}}\right\}$ the event that the $R_{k+1}$ relay is active. Then, we may write:

\begingroup
\begin{eqnarray}
p^{k+1}_{R_i} &{=}& p^k_{\left.R_i\right|\overline{B}_{k+1}}Pr\left\{\overline{B}_{k+1}\right\}+p^k_{\left.R_{i-1}\right|B_{k+1}}p_{B_{k+1}}
\end{eqnarray}
\endgroup

\noindent and $p^{k+1}_{R_{k+1}} {=} p^{k}_{\left.R_k\right|B_{k+1}}p_{B_{k+1}}$. The average number of active relays may be then written as: %$E\left[\left|\mathcal{C}_{K+1}\right|\right] =  \sum^{K+1}_{i=1}iP^{K+1}_{R_i}$ or:

%\vspace{-1pt}

\begingroup

{\setlength{\arraycolsep}{0em}\begin{eqnarray}
&&\mathbf{E}\left[\left|\mathcal{A}_{k+1}\right|\right] =  \sum^{k+1}_{i=1}ip^{k+1}_{R_i}\nonumber\\
&&=\sum^{k}_{i=1}i\left(p^k_{\left.R_i\right|\overline{B}_{k+1}}p_{\overline{B}_{k+1}}+p^k_{\left.R_{i-1}\right|B_{k+1}}p_{B_{k+1}}\right)+(k+1)p^{k+1}_{R_{k+1}}\nonumber\\
&&=\sum^{k}_{i=1}ip^k_{\left.R_i\right|\overline{B}_{k+1}}p_{\overline{B}_{k+1}}+\sum^{k}_{j=0}(j+1)p^k_{\left.R_j\right|B_{k+1}}p_{B_{k+1}}\nonumber\\
&&=\sum^k_{i=1}i\left[p^k_{\left.R_i\right|\overline{B}_{k+1}}p_{\overline{B}_{k+1}}+p^k_{\left.R_i\right|B_{k+1}}p_{B_{k+1}}\right]+\sum^{k}_{i=0}p^k_{\left.R_i\right|B_{k+1}}p_{B_{k+1}}\nonumber\\
&&= \mathbf{E}\left[\left|\mathcal{A}_{k}\right|\right] + p_{B_{k+1}},
\end{eqnarray}}
\endgroup

%\vspace{-1pt}

\noindent where $p_{B_{k+1}} = Pr\left\{\boldsymbol{1}_{AR_{k+1}},\boldsymbol{1}_{BR_{k+1}}\right\}$ is the probability that both $AR_{k+1}$ and $BR_{k+1}$ links are active, given by:

\begin{equation}
p_{B_{k+1}} =  Q\left(\frac{\gamma^* - \mu_{AR_{k+1}}}{\sigma_{AR_{k+1}}}\right)Q\left(\frac{\gamma^* - \mu_{BR_{k+1}}}{\sigma_{BR_{k+1}}}\right).
\end{equation}

\ifCLASSOPTIONcaptionsoff
  \newpage
\fi

\bibliographystyle{IEEEtran}
\bibliography{IEEEabrv,mybibliography}

\begin{thebibliography}{10}
\providecommand{\url}[1]{#1}
\csname url@samestyle\endcsname
\providecommand{\newblock}{\relax}
\providecommand{\bibinfo}[2]{#2}
\providecommand{\BIBentrySTDinterwordspacing}{\spaceskip=0pt\relax}
\providecommand{\BIBentryALTinterwordstretchfactor}{4}
\providecommand{\BIBentryALTinterwordspacing}{\spaceskip=\fontdimen2\font plus
\BIBentryALTinterwordstretchfactor\fontdimen3\font minus
  \fontdimen4\font\relax}
\providecommand{\BIBforeignlanguage}[2]{{%
\expandafter\ifx\csname l@#1\endcsname\relax
\typeout{** WARNING: IEEEtran.bst: No hyphenation pattern has been}%
\typeout{** loaded for the language `#1'. Using the pattern for}%
\typeout{** the default language instead.}%
\else
\language=\csname l@#1\endcsname
\fi
#2}}
\providecommand{\BIBdecl}{\relax}
\BIBdecl

\bibitem{cl}
S.~Shakkottai, T.~Rappaport, and P.~Karlsson, ``Cross-{L}ayer {D}esign for
  {W}ireless {N}etworks,'' \emph{IEEE Communications Magazine}, vol.~41,
  no.~10, pp. 74--80, Oct 2003.

\bibitem{cl5}
X.~Zhang, J.~Tang, H.-H. Chen, S.~Ci, and M.~Guizani, ``Cross-{L}ayer-based
  {M}odeling for {Q}uality of {S}ervice {G}uarantees in {M}obile {W}ireless
  {N}etworks,'' \emph{IEEE Communications Magazine}, vol.~44, no.~1, pp.
  100--106, Jan 2006.

\bibitem{cl1}
L.-C. Wang, A.~Chen, and S.-Y. Huang, ``A {C}ross-{L}ayer {I}nvestigation for
  the {T}hroughput {P}erformance of {CSMA/CA}-based {WLAN}s with {D}irectional
  {A}ntennas and {C}apture {E}ffect,'' \emph{IEEE Transactions on Vehicular
  Technology}, vol.~56, no.~5, pp. 2756--2766, Sept 2007.

\bibitem{shad3}
C.-S. Hwang, K.~Seong, and J.~Cioffi, ``Throughput {M}aximization by
  {U}tilizing {M}ulti-{U}ser {D}iversity in {S}low-{F}ading {R}andom {A}ccess
  {C}hannels,'' \emph{IEEE Transactions on Wireless Communications}, vol.~7,
  no.~7, pp. 2526--2535, 2008.

\bibitem{cl6}
G.~Femenias, J.~Ramis, and L.~Carrasco, ``Using {T}wo-{D}imensional {M}arkov
  {M}odels and the {E}ffective-{C}apacity {A}pproach for {C}ross-{L}ayer
  {D}esign in {AMC/ARQ}-based {W}ireless {N}etworks,'' \emph{IEEE Transactions
  on Vehicular Technology}, vol.~58, no.~8, pp. 4193--4203, Oct 2009.

\bibitem{cross}
V.~Mahinthan, H.~Rutagemwa, J.~W. Mark, and X.~Shen, ``Cross-{L}ayer
  {P}erformance {S}tudy of {C}ooperative {D}iversity {S}ystem with {ARQ},''
  \emph{IEEE Transactions on Vehicular Technology}, vol.~58, no.~2, pp.
  705--719, 2009.

\bibitem{cl7}
H.~Rutagemwa, T.~Willink, and L.~Li, ``Modeling and {P}erformance {A}nalysis of
  {M}ultihop {C}ooperative {W}ireless {N}etworks,'' \emph{IEEE Transactions on
  Vehicular Technology}, vol.~59, no.~6, pp. 3057--3069, July 2010.

\bibitem{shad4}
M.~Di~Renzo, J.~Alonso-Zarate, L.~Alonso, and C.~Verikoukis, ``On the {I}mpact
  of {S}hadowing on the {P}erformance of {C}ooperative {M}edium {A}ccess
  {C}ontrol {P}rotocols,'' in \emph{IEEE Global Telecommunications Conference
  (GLOBECOM)}, Dec. 2011, pp. 1--6.

\bibitem{cl8}
Y.~Zhou, J.~Liu, C.~Zhai, and L.~Zheng, ``Two-{T}ransmitter {T}wo-{R}eceiver
  {C}ooperative {MAC} {P}rotocol: {C}ross-{L}ayer {D}esign and {P}erformance
  {A}nalysis,'' \emph{IEEE Transactions on Vehicular Technology}, vol.~59,
  no.~8, pp. 4116--4127, Oct 2010.

\bibitem{cl2}
H.~Shan, H.~T. Cheng, and W.~Zhuang, ``Cross-{L}ayer {C}ooperative {MAC}
  protocol in {D}istributed {W}ireless {N}etworks,'' \emph{IEEE Transactions on
  Wireless Communications}, vol.~10, no.~8, pp. 2603--2615, August 2011.

\bibitem{cope}
S.~Katti, H.~Rahul, W.~Hu, D.~Katabi, M.~Medard, and J.~Crowcroft, ``{XOR}s in
  the {A}ir: {P}ractical {W}ireless {N}etwork {C}oding,'' \emph{IEEE/ACM
  Transactions on Networking}, vol.~16, no.~3, pp. 497--510, 2006.

\bibitem{argyriou}
A.~Argyriou, ``Wireless {N}etwork {C}oding with {I}mproved {O}pportunistic
  {L}istening,'' \emph{IEEE Transactions on Wireless Communications}, vol.~8,
  no.~4, pp. 2014--2023, 2009.

\bibitem{phoenix}
A.~Munari, F.~Rossetto, and M.~Zorzi, ``Phoenix: {M}aking {C}ooperation more
  {E}fficient through {N}etwork {C}oding in {W}ireless {N}etworks,'' \emph{IEEE
  Transactions on Wireless Communications}, vol.~8, no.~10, pp. 5248--5258,
  2009.

\bibitem{wang}
X.~Wang and J.~Li, ``Network {C}oding {A}ware {C}ooperative {MAC} {P}rotocol
  for {W}ireless {A}d {H}oc {N}etworks,'' \emph{IEEE Transactions on Parallel
  and Distributed Systems}, vol.~25, no.~1, pp. 167--179, 2014.

\bibitem{umehara}
D.~Umehara, S.~Denno, M.~Morikura, and T.~Sugiyama, ``Performance {A}nalysis of
  {S}lotted {ALOHA} and {N}etwork {C}oding for {S}ingle-{R}elay {M}ulti-{U}ser
  {W}ireless {N}etworks,'' \emph{Ad Hoc Networks}, vol.~9, no.~2, pp. 164--179,
  2011.

\bibitem{nccarq}
A.~Antonopoulos, C.~V. Verikoukis, C.~Skianis, and {\"O}.~B. Akan, ``Energy
  {E}fficient {N}etwork {C}oding-based {MAC} for {C}ooperative {ARQ} {W}ireless
  {N}etworks,'' \emph{Ad Hoc Networks}, vol.~11, no.~1, pp. 190--200, 2013.

\bibitem{cor1}
P.~Agrawal and N.~Patwari, ``Correlated {L}ink {S}hadow {F}ading in {M}ulti-hop
  {W}ireless {N}etworks,'' \emph{IEEE Transactions on Wireless Communications},
  vol.~8, no.~8, pp. 4024--4036, 2009.

\bibitem{cor2}
S.~Szyszkowicz, H.~Yanikomeroglu, and J.~Thompson, ``On the {F}easibility of
  {W}ireless {S}hadowing {C}orrelation {M}odels,'' \emph{IEEE Transactions on
  Vehicular Technology}, vol.~59, no.~9, pp. 4222--4236, 2010.

\bibitem{cor3}
A.~Lalos, M.~Di~Renzo, L.~Alonso, and C.~Verikoukis, ``Impact of {C}orrelated
  {L}og-{N}ormal {S}hadowing on {T}wo-{W}ay {N}etwork {C}oded {C}ooperative
  {W}ireless {N}etworks,'' \emph{IEEE Communications Letters}, vol.~17, no.~9,
  pp. 1738--1741, 2013.

\bibitem{cor4}
A.~Antonopoulos, M.~Di~Renzo, and C.~Verikoukis, ``Effect of {R}ealistic
  {C}hannel {C}onditions on the {E}nergy {E}fficiency of {N}etwork
  {C}oding-aided {C}ooperative {MAC} {P}rotocols,'' \emph{IEEE Wireless
  Communications}, vol.~20, no.~5, pp. 76--84, 2013.

\bibitem{thesismary}
P.~Mary, ``Etude {A}nalytique des {P}erformances des {S}ystemes
  {R}adio-{M}obiles en {P}resence d'{E}vanouissements et d'{E}ffet de
  {M}asque,'' 2008.

\bibitem{5288484}
P.~Mary, M.~Dohler, J.-M. Gorce, G.~Villemaud, and M.~Arndt, ``M-ary {S}ymbol
  {E}rror {O}utage over {N}akagami-m {F}ading {C}hannels in {S}hadowing
  {E}nvironments,'' \emph{IEEE Transactions on Communications}, vol.~57,
  no.~10, pp. 2876--2879, Oct. 2009.

\bibitem{104090}
M.~Gudmundson, ``Correlation {M}odel for {S}hadow {F}ading in {M}obile {R}adio
  {S}ystems,'' \emph{Electronics Letters}, vol.~27, no.~23, pp. 2145 --2146,
  Nov. 1991.

\bibitem{toeplitz}
W.~F. Trench, ``Properties of {S}ome {G}eneralizations of
  {K}ac-{M}urdock-{S}zeg{\"o} {M}atrices,'' \emph{Structured Matrices in
  Mathematics, Computer Science and Engineering II}, pp. 233--245, 2001.

\bibitem{80211}
D.~Vassis, G.~Kormentzas, A.~Rouskas, and I.~Maglogiannis, ``The {IEEE} 802.11g
  {S}tandard for {H}igh {D}ata {R}ate {WLAN}s,'' \emph{IEEE Network}, vol.~19,
  no.~3, pp. 21--26, 2005.

\bibitem{book2}
M.~R. Spiegel, \emph{Mathematical Handbook of Formulas and Tables}, ser.
  Schaum's Outline Series.\hskip 1em plus 0.5em minus 0.4em\relax McGraw-Hill,
  1968.

\bibitem{1969method}
N.~M. Steen, G.~D. Byrne, and E.~M. Gelbard, ``Gaussian {Q}uadratures for the
  {I}ntegrals $\int^{\infty}_{0}exp\left[-x^2\right]f(x)dx$ and
  $\int^{b}_{0}exp\left[-x^2\right]f(x)dx$,'' \emph{Mathematics of
  Computation}, vol.~23, no. 107, pp. 661--671, Mar. 1969.

\bibitem{bianchi}
G.~Bianchi, ``Performance {A}nalysis of the {IEEE} 802.11 {D}istributed
  {C}oordination {F}unction,'' \emph{IEEE Journal on Selected Areas in
  Communications}, vol.~18, no.~3, pp. 535--547, 2000.

\bibitem{80211std}
``Draft {IEEE} {S}tandard for {I}nformation {T}echnology {T}elecommunications
  and {I}nformation {E}xchange between {S}ystems {L}ocal and {M}etropolitan
  {A}rea {N}etworks {S}pecific {R}equirements {P}art 11: {W}ireless {M}edium
  {A}ccess {C}ontrol ({MAC}) and {P}hysical {L}ayer ({PHY}) {S}pecifications:
  {A}mendment 6 by {IEEE} {S}td 802.11g-2003 and {IEEE} {S}td 802.11h-2003),''
  \emph{IEEE Std P802.11i/D10.0}, pp.~--, 2004.

\bibitem{ebert}
J.~Ebert, S.~Aier, A.~Kofahl, B.~Becker, B.~Burns, and A.~Wolisz, ``Measurement
  and simulation of the energy consumption of a wlan interface,''
  Telecommunication Networks Group, Technical University of Berlin, Tech. Rep.
  Tech. Rep. TKN-02-010, June 2002.

\bibitem{out_cor}
R.~D. Stevens and I.~Dilworth, ``Mobile {R}adio {S}hadowing {L}oss
  {V}ariability and {C}o-channel {S}ignal {C}orrelation at 452 {MH}z,''
  \emph{Electronics Letters}, vol.~32, no.~1, pp. 16--17, 1996.

\bibitem{in_cor1}
K.~Butterworth, K.~Sowerby, and A.~Williamson, ``Correlated {S}hadowing in an
  {I}n-{B}uilding {P}ropagation {E}nvironment,'' \emph{Electronics Letters},
  vol.~33, no.~5, pp. 420--422, 1997.

\bibitem{in_cor2}
------, ``Base {S}tation {P}lacement for {I}n-{B}uilding {M}obile
  {C}ommunication {S}ystems to {Y}ield {H}igh {C}apacity and {E}fficiency,''
  \emph{IEEE Transactions on Communications}, vol.~48, no.~4, pp. 658--669,
  2000.

\end{thebibliography}

\end{document}